\documentclass[a4paper]{article}

\usepackage{amsmath,amssymb,amsthm}
\usepackage{float}
\usepackage{graphicx}
\usepackage{hyperref}
\usepackage[capitalize,nameinlink]{cleveref}

\newcommand{\Z}{\mathbb{Z}}
\newcommand{\Zp}{\mathbb{Z^+}}
\newcommand{\Q}{\mathbb{Q}}
\newcommand{\s}{\ \!}

\newcommand{\keywords}[1]{{\bf Keywords: }{#1}}
\newcommand{\subclass}[1]{{\bf Mathematics Subject Classification (2020): }{#1}}

\newtheorem{thm}{theorem}[section]
\newtheorem{remark}[thm]{Remark}
\newtheorem{example}[thm]{Example}
\newtheorem{definition}[thm]{Definition}
\newtheorem{theorem}[thm]{Theorem}
\newtheorem{proposition}[thm]{Proposition}
\newtheorem{corollary}[thm]{Corollary}
\newtheorem{lemma}[thm]{Lemma}

\crefname{equation}{Equation}{Equations}
\Crefname{equation}{Equation}{Equations}
\crefname{figure}{Figure}{Figures}
\Crefname{figure}{Figure}{Figures}

\ifpdf
\hypersetup{
  pdftitle={Euclidean Affine Functions and Applications to Calendar Algorithms},
  pdfauthor={Cassio Neri and Lorenz Schneider}
}
\fi

\title{
  Euclidean Affine Functions and Applications to Calendar Algorithms
  \thanks{Submitted to the editors \today}
}

\author{
  Cassio Neri
  \thanks{JPMorgan Chase \& Co, \texttt{cassio.neri@jpmorgan.com} (Opinions
  expressed in this paper are those of the authors, and do not necessarily
  reflect the view of JP Morgan.)}
  \and
  Lorenz Schneider
  \thanks{EMLYON Business School, \texttt{schneider@em-lyon.com}.}
}

\begin{document}

\maketitle

\begin{abstract}
We study properties of Euclidean affine functions (EAFs), namely those of the form
$f(r) = (\alpha\cdot r + \beta)/\delta$, and their closely related
expression $\mathring{f}(r) = (\alpha\cdot r + \beta)\%\delta$, where $r$, $\alpha$, $\beta$ and
$\delta$ are integers, and where $/$ and $\%$ respectively denote the quotient and remainder of
Euclidean division. We derive algebraic relations and numerical approximations
that are important for the efficient evaluation of these expressions in modern CPUs.
Since simple division and remainder are particular cases of EAFs (when $\alpha = 1$ and $\beta
= 0$), the optimisations proposed in this paper can also be appplied to them. Such
expressions appear in some of the most common tasks in any computer system, such as
printing numbers, times and dates. We use calendar calculations as the main
application example because it is richer with respect to the number of
EAFs employed. Specifically, the main application
presented in this article relates to Gregorian calendar algorithms. We will show
how they can be implemented substantially more efficiently than is currently
the case in widely used C, C++, C\# and Java open source libraries. Gains in speed of a
factor of two or more are common.
\end{abstract}

\keywords{Euclidean Affine Functions, Integer Division, Calendar Algorithms}

\subclass{65Y04, 68R01}

\section{Introduction}
\label{sec:Introduction}

Divisions by constants appear in some of the most common tasks performed by
software systems, including printing decimal numbers (division by powers of $10$) and
working with times (division by $24$, $60$ and $3600$). Since division is the slowest of
the four basic arithmetical operations, various authors \cite{Alverson1991,
CavagninoWerbrouck2007,GranlundMontgomery1994,MagenheimerPetersPettisZuras1987,
Robison2005} have proposed strength reduction optimisations ({\em i.e.,}
the replacement of instructions with alternatives that are mathematically equivalent
but faster) for integer divisions when divisors are constants known by the
compiler. The algorithms proposed by Granlund and Montgomery
\cite{GranlundMontgomery1994} have been implemented by major compilers.

A typical example that will appear later in this article is the division
\begin{equation*}
n / 1461 = 2\s939\s745\cdot n/2^{32}, \quad \forall n\in[0, 28\s825\s529[.
\end{equation*}
The problem involves finding, for the given divisor $1461$,
the multiplier 2\s939\s745, the exponent $32$, and the boundary $28\s825\s529$ of the interval
in which equality holds. Clearly, there are many possible solutions.

Remainder calculation is a closely related problem, also arising in the tasks
mentioned above, although the cited works do not consider optimisations for this operation.
Compilers are content to apply strength reduction to obtain the quotient
$r/\delta$ and the remainder evaluating $r - \delta\cdot(r/\delta)$. For
this reason, \cite{LemireKaserKurz2019, Warren2013} considered the problem of
directly obtaining the remainder without first calculating the quotient.

This paper expands previous works in two ways. Firstly, it considers the more
general setting of Euclidean affine functions (EAFs) $f(r) = (\alpha\cdot r + \beta)/
\delta$ where $\alpha$, $\beta$ and $\delta$ are integer constants and $r$ is an
integer variable. Secondly, it suggests optimisations that are even more
effective in applications where both the quotient and the remainder need to be
evaluated.

We derive EAF-related equalities that provide alternative ways of evaluating
expressions commonly used in applications. (For instance, $r - \delta\cdot(r/
\delta)$, which is used to obtain the remainder as explained above.) These equalities underpin
optimisations, other than strength reduction, that take into account aspects of
modern CPUs. Specifically, they foster instruction-level parallelism implemented
by superscalar processors and they profit from the backward compatibility features
that drove the design of the x86\_64 instruction set.

Calendar calculations are a rich application field for our EAF results.
Performed by various software systems, they tackle questions such as: What is the
date today? For how many days do interest rates accrue over the period of a loan?
Your mobile phone no doubt performed a similar calculation while you read this
paragraph.

Our algorithms are substantially faster than those in widely used open source
implementations, as shown by benchmarks against counterparts in glibc
\cite{GlibcMktime, GlibcOfftime} (Linux contains a similar implementation
\cite{LinuxTimeconv}), Boost \cite{BoostGregorianCalendar}, libc++
\cite{LLVMChrono}, .NET \cite{DotNetDateTime} and OpenJDK
\cite{OpenJDKLocalDate} (Android contains the same implementations
\cite{AndroidLocalDate}). Our algorithms are also faster than others found in the academic
literature \cite{Baum1998, FliegelFlandern1968, Hatcher1984, Hatcher1985,
ReingoldDershowitz2018, Richards1998}.

The critical point setting apart implementations of calendar algorithms is how
they deal with non-linearities caused by irregular month and year lengths. Some authors
\cite{LinuxTimeconv, GlibcOfftime, DotNetDateTime} resort to look-up tables,
which can be costly when the L1-level cache is cold. They conduct linear searches on
these tables, entailing branching that can cause stalls in the processor's
execution pipeline. Others \cite{Baum1998, FliegelFlandern1968, Hatcher1984,
Hatcher1985, BoostGregorianCalendar, LLVMChrono, ReingoldDershowitz2018,
Richards1998} tackle the issue entirely through EAFs. Nevertheless, they do not
go far enough and do not use the mathematical properties derived in this paper.
To the best of our knowledge, the most successful attempt at achieving this efficiency is Baum
\cite{Baum1998}. This author provides some explanations, but appears to resort to trial
and error when it comes to pre-calculating certain {\em magic numbers} used in his
algorithm.
We go further by providing a systematic and general framework for such calculations.

Our paper and its main contributions are organised as follows.

\Cref{sec:EAF} introduces EAFs and derives some of
their properties. Although they are common in applications, we were unable to
find any systematic coverage of them. Hatcher \cite{Hatcher1984} and Richards
\cite{Richards1998} appear to be aware of some of the results of \cref{thm:EAF}, but
they do not point to any proof.

\Cref{sec:FastEAF} concerns efficient evaluations of EAFs. \cref{thm:FastEAFUp,%
thm:FastEAFDown} generalise prior-known results on division to EAFs.
\Cref{sec:SpecialCase} revisits those results and brings geometric insights to
the problem. \Cref{thm:FastResidual} concerns the efficient evaluation of residuals:
they are to EAFs what remainders are to division. The optimisation proposed by
\cref{thm:FastResidual} is fundamentally different from other optimisations, both in the present paper and in prior
works, since it does not involve strength reduction. Indeed, this theorem shows
how to break a data dependency present in the instructions currently emitted by
compilers, enabling instruction-level parallelism, as we shall see in
\cref{ex:FastResidualYear}.

\Cref{sec:Calendars} provides a generic mathematical framework for calendars,
which we use to study the Gregorian calendar. Efficient algorithms for this
calendar are derived in \cref{sec:RataDie,sec:RataDieGregorian}.
\Cref{sec:Performance} presents performance analysis.

We finish this introduction by setting the notation used throughout and
recalling some well-known results.

On the totally ordered Euclidean domain of integer numbers $(\Z, +, \cdot,
\le)$, forward slash $/$ and percent $\%$ respectively denote the {\bf quotient} and {\bf
remainder} of Euclidean division. More precisely, given $n,\delta\in\Z$ with
$\delta\ne 0$, there exist unique $q,r\in\Z$, with $0\le r < |\delta|$, such
that $n = q\cdot\delta + r$. Then $n/\delta = q$ and $n\%\delta = r$. (These
concepts, for non-negative operands, match the usual operators \texttt{/} and
\texttt{\%} of many programming languages in the C-family.)

We follow the usual algebraic order of operation rules and, in accordance with
C-style languages, we give $\%$ the same precedence as $\cdot$ and $/$. Hence,
$a\cdot b\%c = (a\cdot b)\%c$, $a\%b\cdot c = (a\%b)\cdot c$ and $a/b\%c = (a/b)
\%c$, $a\%b/c = (a\%b)/c$. Moreover, $a + b\%c = a + (b\%c)$, and $a - b\%c = a -
(b\%c)$. In general, $a\cdot b/c\ne a\cdot(b/c)$ and $(a + b)/c\ne a/c + b/c$.
This indicates that these precedence rules are even more important than when
working in algebraic fields. Nevertheless, we shall drop unnecessary parentheses
henceforth.

The set of non-negative integer numbers is denoted by $\Zp = \{ x\in \Z\ ;\ x\ge
0\}$.

On the totally ordered field of rational numbers $(\Q, +, \cdot, \le)$, the
reciprocal of $\delta\in\Q\setminus\{0\}$ is denoted by $\delta^{-1}$. For $n$,
$\delta\in\Z$ with $\delta\ne 0$, we shall write $n\cdot\delta^{-1}\in\Q$ to
distinguish it from $n/\delta\in\Z$. If the former belongs to $\Z$, then $n\cdot
\delta^{-1} = n/\delta$. This distinction is particularly important when dealing
with inequalities. On the one hand, for $x, y, \delta\in\Z$ with $\delta > 0$,
we have $x\cdot\delta^{-1}\ge y\cdot\delta^{-1}$ if, and only if, $x\ge y$;
and on the another hand $x/\delta\ge y/\delta$ if $x\ge y$, but the converse does not
hold.

For any $x\in\Q$, there exist unique $\left\lfloor x\right\rfloor\in\Z$ and
$\left\lceil x\right\rceil\in\Z$ such that $\left\lfloor x\right\rfloor\le x <
\left\lfloor x\right\rfloor + 1$ and $\left\lceil x\right\rceil - 1 < x \le\left
\lceil x\right\rceil$. For any finite set $X$, the {\bf number of elements} of
$X$ is denoted by $\#X$.

\section{Euclidean affine functions}
\label{sec:EAF}

Multiplication by $60$ converts hours to minutes. Conversely, division by $60$
converts minutes to hours. If the amount of minutes in an hour were variable, these
problems would be more complex. Since months and years have variable numbers of
days, such complexities arise in calendar calculations. Different ways of
tackling this problem appear across implementations with various degrees of
performance. Nevertheless, similarities between calendar calculations and the
previous linear problems are stronger than they might appear, as revealed by our
study of EAFs.

\begin{definition}
\label{def:EAF}
A function $f:\Z\rightarrow\Z$ is a {\bf Euclidean affine function}, {\bf EAF}
for short, if it has the form $f(r) = (\alpha\cdot r + \beta)/\delta$ for all $r
\in\Z$ and fixed $\alpha$, $\beta$, $\delta\in\Z$ with $\delta\ne0$.
\end{definition}

This terminology is ours. The analogues of EAFs in higher dimensions appear in
Discrete Geometry and are called quasi-affine transformations. That area focuses
on periodicity, tiling and other geometric aspects, whereas we are concerned with
efficient calculations. We therefore use the term EAF to distinguish between these two approaches.

\begin{definition}
\label{def:Residual}
Let $f$ be the EAF $f(r) = (\alpha\cdot r + \beta)/\delta$. The {\bf residual
function} of $f$ is $\mathring{f}(r) = (\alpha\cdot r + \beta)\%\delta$ or,
equivalently, $\mathring{f}(r) = \alpha\cdot r + \beta - \delta\cdot f(r)$.
\end{definition}

An important function related to $f(r) = r/\delta$ is its right inverse $\hat{f}
(q) = \delta\cdot q$. If $\delta>0$ and $r\in[\hat{f} (q), \hat{f}(q + 1)[\ =
[\delta\cdot q, \delta\cdot(q + 1)[$, then $f(r) = q$ and $\mathring{f}(r) = r -
\delta\cdot(r/\delta) = r - \hat{f}(f(r)) = r - \hat{f}(q)$. As we shall see,
\cref{thm:EAF} generalises these results and supports the following terminology.

\begin{definition}
\label{def:MRI}
Let $f$ be the EAF $f(r) = (\alpha\cdot r + \beta)/\delta$, with $\delta\ge
\alpha > 0$. The {\bf minimal right inverse} of $f$ is the EAF $\hat{f}(q) =
(\delta\cdot q + \alpha - \beta - 1)/\alpha$.
\end{definition}

\begin{lemma}
\label{lem:Residual}
Let $f$ be the EAF $f(r) = (\alpha\cdot r + \beta)/\delta$, with $\delta > 0$,
and $g:\Z\rightarrow\Z$. For any $q\in\Z$ such that $f(g(q) - 1) < f(g(q))$, we
have:
\begin{equation*}
r\in\Z \quad\text{and}\quad f(r) = f(g(q)) \quad\Longrightarrow\quad
\mathring{f}(r)/\alpha = r - g(q).
\end{equation*}
\end{lemma}

\begin{proof}
Since $\mathring{f}(r) = \alpha\cdot r + \beta - \delta\cdot f(r)$, we have
$\mathring{f}(g(q)) = \alpha\cdot g(q) + \beta - \delta\cdot f(g(q))$. Hence,
\begin{align*}
\mathring{f}(r)
&= \alpha\cdot r + \beta - \delta\cdot f(r)
 = \alpha\cdot (r - g(q)) + \alpha\cdot g(q) + \beta - \delta\cdot f(g(q)) \\
&= \alpha\cdot (r - g(q)) + \mathring{f}(g(q)).
\end{align*}
We will show that $0\le\mathring{f}(g(q)) < \alpha$ and it will follow that
$\mathring{f}(r)/\alpha = r - g(q)$. By definition of $\mathring{f}$ and $\%$ we
obtain $\mathring{f}(g(q))\ge 0$.
Suppose, by contradiction, that $\mathring{f}(g(q))\ge\alpha$. Then,
\begin{align*}
\alpha\cdot(g(q) - 1) + \beta
&= \alpha\cdot g(q) + \beta - \alpha
 = \delta\cdot f(g(q)) + \mathring{f}(g(q))- \alpha \\
&\ge \delta\cdot f(g(q)).
\end{align*}
Dividing the above by $\delta$ gives $f(g(q) - 1)\ge f(g(q))$, which contradicts
the assumption on $q$.
\end{proof}

\begin{theorem}
\label{thm:EAF}
Let $f$ be the EAF $f(r) = (\alpha\cdot r + \beta)/\delta$ with $\delta\ge\alpha
> 0$. Then, for any $r$ and $q\in\Z$ we have:
\begin{enumerate}
\item\label{it:RInverse} $f(\hat{f}(q)) = q$ and $f(\hat{f}(q) - 1) = q - 1$;

\item\label{it:Interval} $f(r) = q$ if, and only if, $r\in[\hat{f}(q), \hat{f}(q
+ 1)[$;

\item\label{it:Residual} $r\in[\hat{f}(f(r)), \hat{f}(f(r) + 1)[$ and $\mathring
{f}(r)/\alpha = r - \hat{f}(f(r))$.
\end{enumerate}
\end{theorem}

\begin{proof}
Let $q\in\Z$ and set $n = \delta\cdot q + \alpha - \beta - 1$ so that $\hat{f}
(q) = (\delta\cdot q + \alpha - \beta - 1)/\alpha = n/\alpha$.

\ref{it:RInverse}:
Since $0\le n\%\alpha\le\alpha - 1$ and $\alpha\le\delta$, we have:
\begin{equation}
\label{eq:alphadelta}
0\le\alpha - 1 - n\%\alpha\le\delta - 1 - n\%\alpha < \delta.
\end{equation}
From $\hat{f}(q) = n/\alpha$, we obtain $\alpha\cdot\hat{f}(q) = n - n\%\alpha$
and
\begin{equation*}
\alpha\cdot \hat{f}(q) + \beta = n + \beta - n\%\alpha
= \delta\cdot q + \alpha - 1 - n\%\alpha.
\end{equation*}
This and \cref{eq:alphadelta} yield $(\alpha\cdot \hat{f}(q) + \beta)/\delta =
q$, that is, $f(\hat{f}(q)) = q$. Furthermore, the equality above gives:
\begin{equation*}
\alpha\cdot(\hat{f}(q) - 1) + \beta = \delta\cdot q  - 1 - n\%\alpha
= \delta\cdot(q - 1) + \delta - 1 - n\%\alpha.
\end{equation*}
This and \cref{eq:alphadelta} yield $(\alpha\cdot(\hat{f}(q) - 1) + \beta)/
\delta = q - 1$, that is, $f(\hat{f}(q) - 1) = q - 1$.

\ref{it:Interval}: Since $\delta\ge\alpha > 0$, $f$ is non-decreasing. If $r
\le\hat{f}(q) - 1$, then $f(r)\le f(\hat{f}(q) - 1) = q - 1$. If $\hat{f}(q)\le
r\le\hat{f}(q + 1) - 1$, then $q = f(\hat{f}(q))\le f(r)\le f(\hat{f}(q + 1) -
1) = q$ (the last equality follows from \labelcref{it:RInverse} applied to $q + 1$)
and thus, $f(r) = q$. Finally, if $r\ge\hat{f}(q + 1)$, then $f(r)\ge f(\hat{f}
(q + 1)) = q + 1$ (again from \labelcref{it:RInverse} applied to $q + 1$).

\ref{it:Residual}: Since $\delta\ge\alpha>0$, $\hat{f}$ is strictly
increasing and there thus exists $p\in\Z$ such that $r\in[\hat{f}(p), \hat{f}(p
+ 1)[$. Using \ref{it:RInverse} and \ref{it:Interval} (for $q = p$) gives
$f(\hat{f}(p)) = p$, $f(\hat{f}(p) - 1) = p - 1$ and $f(r) = p$. In particular,
$f(\hat{f}(p) - 1) < f(\hat{f}(p))$ and $f(r) = f(\hat{f}(p))$.
\cref{lem:Residual} (for $g = \hat{f}$ and $q = p$) therefore gives $\mathring{f}(r)/\alpha
= r - \hat{f}(p)$. Using $p = f(r)$ again yields $r\in[\hat{f}(f(r)),
\hat{f}(f(r) + 1)[$ and $\mathring{f}(r)/\alpha = r - \hat{f}(f(r))$.
\end{proof}

From \ref{it:RInverse}, $\hat{f}$ is a right inverse of $f$. Item \ref{it:Interval}
goes further, showing the minimum solution $r$ for the equation $f(r) = q$ is $r
= \hat{f}(q)$, hence the wording of \cref{def:MRI}.

\begin{example}
\label{ex:ResidualCentury}
\Cref{thm:EAF} for $f(r_0) = (4\cdot r_0 + 3)/146097$ and $\hat{f}(q) = 146097
\cdot q/4$ gives:
\begin{equation*}
r_0\in[\hat{f}(q), \hat{f}(q + 1)[ \quad\Longrightarrow\quad (4
\cdot r_0 + 3)/146097 = q \quad\text{and}\quad (4\cdot r_0 + 3)\%146097/4 = r_0
- \hat{f}(q).
\end{equation*}
Each side of last equality provides a way to evaluate the same quantity. To
analyse performance trade-offs, we introduce auxiliary variables referring to
partial results and compare both methods side-by-side:
\begin{equation*}
\begin{aligned}
  & \text{ Calculation of } (4\cdot r_0 + 3)\% 146097/4: \\
n &= 4\cdot r_0 + 3, \\
q &= n/146097, \\
u &= n - 146097\cdot q, \quad (\text{\em i.e., } u = n \% 146097) \\
r &= u/4, \\
\end{aligned}
\qquad\vrule\qquad
\begin{aligned}
  & \text{ Calculation of } r_0 - \hat{f}(q): \\
n &= 4\cdot r_0 + 3,  \\
q &= n/146097,        \\
v &= 146097\cdot q/4 \quad (\text{\em i.e., } v = \hat{f}(q)), \\
r &= r_0 - v.
\end{aligned}
\end{equation*}
For fairness of comparison, $u$ is set to $n - 146097\cdot q$ rather than $n\%
146097$, since the former better represents the CPU instructions typically emitted by
compilers to evaluate the latter. Both listings show exactly the same
operations, the only difference being the order of the last two operations. Our preference
for the left listing will be explained in the next steps of the Gregorian calendar algorithm.
Indeed, the expressions above appear in these calculations and are followed
by the evaluation of $4\cdot r + 3$. If $r$ is calculated as shown on the left,
then $4\cdot r + 3 = 4\cdot(u/4) + 3$ and smart compilers reduce this expression
to $u\ |\ 3$, where $|$ denotes bitwise {\em or}.
\end{example}

\begin{example}
\label{ex:ResidualYear}
\Cref{thm:EAF} for $f(r_1) = (4\cdot r_1 + 3)/1461$ and $\hat{f}(q) = 1461\cdot
q/4$ gives:
\begin{equation*}
r_1\in[\hat{f}(q), \hat{f}(q + 1)[ \quad\Longrightarrow\quad (4\cdot r_1 + 3)
/1461 = q \quad\text{and}\quad (4\cdot r_1 + 3)\%1461/4 = r_1 - \hat{f}(q).
\end{equation*}
\end{example}

\begin{example}
\label{ex:ResidualMonth}
\cref{thm:EAF} for $f(r_2) = (5\cdot r_2 + 461)/153$ and $\hat{f}(m_0) = (153
\cdot m_0 - 457)/5$ gives:
\begin{equation*}
r_2\in[\hat{f}(m_0), \hat{f}(m_0 + 1)[ \quad\Longrightarrow\quad (5\cdot r_2 +
461)/153 = m_0 \quad\text{and}\quad (5\cdot r_2 + 461)\%153/5 = r_2
- \hat{f}(m_0).
\end{equation*}
\end{example}

\section{Fast evaluation of Euclidean affine functions}
\label{sec:FastEAF}

This section covers optimisations for $(\alpha\cdot r + \beta)/\delta$. The
particular case $\alpha = 1$ and $\beta = 0$ has been considered by many authors
\cite{Alverson1991, CavagninoWerbrouck2007, GranlundMontgomery1994,
MagenheimerPetersPettisZuras1987, Robison2005}, who have derived strength reductions
whereby $r/\delta$ is reduced to a multiplication and cheaper operations. Major
compilers implement the algorithms of \cite{GranlundMontgomery1994}. Faster
algorithms exist but compilers are not able to use them. For instance, $r$ might
be restricted to a small interval but, unaware of this fact, the compiler must
assume that $r$ can take any allowed value for its type, usually in an interval of
form $[0, 2^w[$. We change focus and instead of looking for an
algorithm on a given interval, we start with the best algorithm we know and
search the largest interval on which it can be applied. If such an interval is satisfactory
for our application, then we use this algorithm.

When $\delta$ is a power of two, the evaluation of $r/\delta$ reduces to a
bitwise shift, which is very cheap in binary CPUs. The Fundamental Theorem of
Arithmetic implies that $\delta$ is a power of two if, and only if, $2^k\%\delta =
0$ for some $k\in\Zp$. For $(\alpha\cdot r + \beta)/\delta$, a trivial reduction
is also available when $2^k\cdot\alpha\%\delta = 0$. Indeed, we have $2^k\cdot
\alpha = \delta\cdot(2^k\cdot\alpha/\delta)$, and setting $\alpha' = 2^k\cdot
\alpha/ \delta$ and $\beta' = 2^k\cdot\beta/\delta$ gives $2^k\cdot\alpha =
\delta\cdot\alpha'$ and
\begin{equation*}
(\alpha\cdot r + \beta)/\delta
= \left(2^k\cdot\alpha\cdot r + 2^k\cdot\beta\right)/\left(2^k\cdot\delta\right)
= \left(\delta\cdot\alpha'\cdot r + 2^k\cdot\beta\right)/\left(2^k\cdot\delta
  \right)
= \left(\alpha'\cdot r + \beta'\right)/2^k.
\end{equation*}

Conceptually, to evaluate $(\alpha\cdot r + \beta)/\delta$ we multiply $r$ by
$\alpha\cdot\delta^{-1}$ and add the result to $\beta\cdot\delta^{-1}$. Assume,
for the time being, that $\beta = 0$. We take an approximation $\alpha'\in\Z$ of
$2^k\cdot\alpha\cdot\delta^{-1}$, where $k\in\Zp$ is carefully chosen, and
evaluate $\alpha'\cdot r/2^k$. Two natural choices for $\alpha'$ are $\left
\lceil2^k\cdot\alpha\cdot\delta^{-1}\right\rceil$ and $\left\lfloor2^k\cdot
\alpha\cdot\delta^{-1}\right\rfloor$, covered by
\cref{thm:FastEAFUp,thm:FastEAFDown}, respectively.

\begin{lemma}
\label{lem:FastEAF}
Let $f$ be the EAF $f(r) = (\alpha\cdot r + \beta)/\delta$. Then,
\begin{equation*}
\alpha\cdot(r/\delta) = f(r) - f(r\%\delta), \qquad\forall r\in\Z.
\end{equation*}
\end{lemma}

\begin{proof}
$f(r) = (\alpha\cdot(\delta\cdot(r/\delta) + r\%\delta) + \beta)/\delta
= \alpha\cdot(r/\delta) + (\alpha\cdot(r\%\delta) + \beta)/\delta = \alpha\cdot
(r/\delta) + f(r\%\delta)$.
\end{proof}

The next theorem does not assume $2^k\cdot\alpha\%\delta\ne 0$, but when
equality holds a simpler reduction can be made as seen above. The result becomes more interesting
when we do have $2^k\cdot\alpha\%\delta\ne 0$, in which case $\left\lceil2^k\cdot\alpha
\cdot\delta^{-1}\right\rceil= 2^k\cdot\alpha/\delta + 1$.

\begin{theorem}
\label{thm:FastEAFUp}
Let $k\in\Zp$ and $f$ be the EAF $f(r) = (\alpha\cdot r + \beta)/\delta$ with
$\delta > 0$. Set $\alpha' = 2^k\cdot\alpha/\delta + 1$, $\varepsilon = \delta -
2^k\cdot\alpha\%\delta$ and $\beta' = - \min\{\alpha'\cdot r - 2^k\cdot f(r)\ ;
\ r\in[0, \delta[\ \}$. For $r\in[0, \delta[$ define:
\begin{equation*}
q(r) = \min\{p\in\Zp\ ; \ \varepsilon\cdot p + \alpha'\cdot r + \beta' - 2^k
\cdot f(r)\ge 2^k\} \quad\text{and}\quad
M(r) = \delta\cdot q(r) + r. \\
\end{equation*}
Let $N = \min\{ M(r)\ ;\ r\in[0, \delta[\ \}$. Then,
\begin{equation*}
(\alpha\cdot r + \beta)/\delta = \left(\alpha'\cdot r + \beta'\right)/2^k,
\quad\forall r\in[0, N[.
\end{equation*}
\end{theorem}

\begin{proof}
First note that $q(r)$ is well-defined since $\varepsilon > 0$. So are $M(r)$
and $N$. Also, $\delta > 0$ implies $M(r)\ge 0$ and, consequently, $N\ge 0$.
Nevertheless, we do not exclude the possibility that $N = 0$, in which case this
theorem's conclusion is vacuously true. The sequel assumes that $N > 0$ and $r\in\Z$.
Note that:
\begin{equation}
\label{eq:epsilonUp}
\varepsilon = \delta - \left(2^k\cdot\alpha - \delta\cdot\left(2^k\cdot\alpha/
\delta\right)\right) = \delta\cdot\left(1 + 2^k\cdot\alpha/\delta\right) - 2^k
\cdot\alpha = \delta\cdot\alpha' - 2^k\cdot\alpha.
\end{equation}
It follows that
$\delta\cdot\alpha'\cdot(r/\delta) - 2^k\cdot\alpha\cdot(r/\delta) = \varepsilon
\cdot(r/\delta)$. Hence,
\begin{equation*}
\alpha'\cdot r - 2^k\cdot\alpha\cdot(r/\delta)
 = \alpha'\cdot(\delta\cdot(r/\delta) + r\%\delta) - 2^k\cdot\alpha\cdot(r/
 \delta) = \varepsilon\cdot(r/\delta) + \alpha'\cdot(r\%\delta).
\end{equation*}
From the above and \cref{lem:FastEAF} we conclude that:
\begin{equation}
\label{eq:FastEAFUp}
\alpha'\cdot r - 2^k\cdot f(r) = \varepsilon\cdot(r/\delta) + \alpha'\cdot(r\%
\delta) - 2^k\cdot f(r\%\delta).
\end{equation}

Suppose further that $r\in[0, N[$. By definition, $N\le M(r\%\delta)$ and hence
$r < M(r\%\delta)$, {\em i.e.}, $\delta\cdot(r/\delta) + r\% \delta < \delta
\cdot q(r\%\delta) + r\%\delta$. It follows that $r/\delta < q(r\%\delta)$. From
the definition of $q(r\%\delta)$ and the fact that $r/\delta\ge 0$, we obtain:
\begin{equation}
\label{eq:small}
\varepsilon\cdot(r/\delta) + \alpha'\cdot(r\%\delta) + \beta' - 2^k\cdot f(r\%
\delta) < 2^k.
\end{equation}

The definition of $\beta'$ yields $\alpha'\cdot(r\%\delta) + \beta' - 2^k\cdot
f(r\%\delta)\ge 0$. This and $\varepsilon\cdot(r/\delta)\ge 0$ mean that the
left side of \cref{eq:small} is non-negative. Therefore, \cref{eq:FastEAFUp}
yields $0\le\alpha'\cdot r + \beta' - 2^k\cdot f(r) < 2^k$ and the conclusion
follows from Euclidean division by $2^k$.
\end{proof}

\begin{example}
\label{ex:FastEAFUpMonthCount}
\Cref{thm:FastEAFUp} for $k = 5$ and $f(m_0) = (153\cdot m_0 - 457)/5$ yields
$\alpha' = 980$, $\beta' = -2928$ and $N = 12$, that is,
\begin{equation*}
(153\cdot m_0 - 457)/5 = (980\cdot m_0 - 2928)/2^5, \quad\forall m_0\in[0, 12[.
\end{equation*}
\end{example}

The next theorem does assume that $2^k\cdot\alpha\%\delta\ne 0$ and thus, $2^k\cdot
\alpha/\delta = \left\lfloor2^k\cdot\alpha\cdot\delta^{-1}\right\rfloor$.

\begin{theorem}
\label{thm:FastEAFDown}
Let $k\in\Zp$ and $f$ be the EAF $f(r) = (\alpha\cdot r + \beta)/\delta$ with
$\delta > 0$ and $2^k\cdot\alpha\%\delta > 0$. Set $\alpha' = 2^k\cdot\alpha/
\delta$, $\varepsilon = 2^k\cdot\alpha\%\delta$ and $\beta' = \min\{2^k - 1 -
(\alpha'\cdot r - 2^k\cdot f(r))\ ; \ r\in[0, \delta[\ \}$. For $r\in[0,
\delta[$ define:
\begin{equation*}
q(r) = \min\{p\in\Zp\ ;\ -\varepsilon\cdot p + \alpha'\cdot r + \beta' - 2^k
\cdot f(r) < 0\} \quad\text{and}\quad
M(r) = \delta\cdot q(r) + r. \\
\end{equation*}
Let $N = \min\{ M(r)\ ;\ r\in[0, \delta[\ \}$. Then,
\begin{equation*}
(\alpha\cdot r + \beta)/\delta = (\alpha'\cdot r + \beta')/2^k,
\quad\forall r\in[0, N[.
\end{equation*}
\end{theorem}

\begin{proof}
The assumption $2^k\cdot\alpha\%\delta > 0$ means that $\varepsilon > 0$ and ensures
that $q(r)$ is well-defined, as are $M(r)$ and $N$. Also, $\delta > 0$ implies
$M(r)\ge 0$ and, consequently, $N\ge 0$. Nevertheless, we do not exclude the
possibility that $N = 0$, in which case this theorem's conclusion is vacuously
true. The sequel assumes that $N > 0$ and $r\in\Z$.

Note that:
\begin{equation}
\label{eq:epsilonDown}
\varepsilon = 2^k\cdot\alpha\%\delta = 2^k\cdot\alpha - \delta\cdot\left(
2^k\cdot\alpha/\delta\right) = 2^k\cdot\alpha - \delta\cdot\alpha'.
\end{equation}
It follows that $2^k\cdot\alpha\cdot(r/\delta) - \delta\cdot\alpha'\cdot(r/
\delta) = \varepsilon \cdot(r/\delta)$. Hence,
\begin{equation*}
\alpha'\cdot r - 2^k\cdot\alpha\cdot(r/\delta)
 = \alpha'\cdot(\delta\cdot(r/\delta) + r\%\delta) - 2^k\cdot\alpha\cdot(r/
\delta) = - \varepsilon\cdot(r/\delta) + \alpha'\cdot(r\%\delta).
\end{equation*}
From the above and \cref{lem:FastEAF} we conclude that:
\begin{equation}
\label{eq:FastEAFDown}
\alpha'\cdot r - 2^k\cdot f(r) = -\varepsilon\cdot(r/\delta) + \alpha'\cdot(r\%
\delta) - 2^k\cdot f(r\%\delta).
\end{equation}

Suppose further that $r\in[0, N[$. By definition $N\le M(r\%\delta)$ and hence,
$r < M(n\%\delta)$, {\em i.e.}, $\delta\cdot(r/\delta) + r\%\delta < \delta\cdot
q(r\%\delta) + r\%\delta$. Therefore, $r/\delta < q(r\%\delta)$. From the
definition of $q(r\%\delta)$ and the fact that $r/\delta\ge 0$, we obtain:
\begin{equation}
\label{eq:small_}
-\varepsilon\cdot(r/\delta) + \alpha'\cdot(r\%\delta) + \beta' - 2^k\cdot f(r\%
\delta)\ge 0.
\end{equation}

The definition of $\beta'$ yields $\alpha'\cdot(r\%\delta) + \beta' - 2^k\cdot
f(r\%\delta)\le 2^k - 1 < 2^k$. This and $\varepsilon\cdot(r/\delta)\ge 0$ mean
that the left side of \cref{eq:small_} is less than $2^k$. Therefore,
\cref{eq:FastEAFDown} yields $0\le\alpha'\cdot r + \beta' - 2^k\cdot f(r) < 2^k$
and the conclusion follows from Euclidean division by $2^k$.
\end{proof}

\begin{example}
\label{ex:FastEAFDownMonthCount}
\Cref{thm:FastEAFDown} for $k = 5$ and $f(m_0) = (153\cdot m_0 - 457)/5$ yields
$\alpha' = 979$, $\beta' = -2919$ and $N = 34$, namely:
\begin{equation}
\label{eq:FastEAFDownMonthCount}
(153\cdot m_0 - 457)/5 = (979\cdot m_0 - 2919)/2^5, \quad\forall m_0\in[0, 34[.
\end{equation}
\end{example}

\cref{ex:FastEAFUpMonthCount,ex:FastEAFDownMonthCount} give two fast
alternatives for the same EAF. In this case, the latter has the advantage of being valid on a
larger range, $[0, 34[$ as opposed to $[0, 12[$, but in other cases the opposite is
true. For EAFs known prior to compilation, one can find both alternatives and select
the one that is valid in the larger range. Compilers can do the same for EAFs found in
source code. However, if the costs of finding the two alternatives need to be
avoided, then a simple and fast-selecting criterion is based on the following
heuristics. To maximise $N$ we seek to maximise $q(r)$, which is the minimum
value of $p$ for which $\varepsilon\cdot p$ reaches a certain threshold. The
smaller $\varepsilon$ is, the larger $p$ must be for this to happen. Hence, we
choose the alternative with the smallest $\varepsilon$ amongst its two possible
values, namely, $\varepsilon_1 = \delta - 2^k\cdot \alpha\%\delta$ and
$\varepsilon_2 = 2^k\cdot\alpha\%\delta$. (If $\delta$ is odd, then
$\varepsilon_1\ne\varepsilon_2$.) Another interpretation of this criterion is
that it sets $\alpha'$ to the best approximation of $2^k\cdot\alpha\cdot\delta^
{-1}$ given by $\left\lceil2^k\cdot\alpha\cdot\delta^{-1}\right\rceil$ or $\left
\lfloor2^k\cdot\alpha\cdot\delta^{-1}\right\rfloor$. Indeed, if $2^k\cdot\alpha
\%\delta\ne 0$, then \cref{eq:epsilonUp,eq:epsilonDown} give:
\begin{equation*}
\begin{aligned}
\varepsilon_1
&= \delta\cdot\left(2^k\cdot\alpha/\delta + 1\right) - 2^k\cdot\alpha
  \hspace{-1.5ex}
&= \delta\cdot\left\lceil 2^k\cdot\alpha\cdot\delta^{-1}\right\rceil - 2^k\cdot
  \alpha,
\\
\varepsilon_2
&= 2^k\cdot\alpha - \delta\cdot\left(2^k\cdot\alpha/\delta\right)
&= 2^k\cdot\alpha - \delta\cdot\left\lfloor 2^k\cdot\alpha\cdot\delta^{-1}\right
  \rfloor.
\end{aligned}
\end{equation*}
It follows that $\varepsilon_1 < \varepsilon_2$ if, and only if, $\left\lceil
2^k\cdot\alpha\cdot\delta^{-1}\right\rceil - 2^k\cdot\alpha\cdot\delta^{-1} <
2^k\cdot\alpha\cdot\delta^{-1} - \left\lfloor 2^k\cdot\alpha\cdot\delta^{-1}
\right\rfloor$.

\begin{example}
\label{ex:FastEAFDownMonth}
\cref{thm:FastEAFDown} for $k = 16$ and $f(r_2) = (5\cdot r_2 + 461)/153$ yields
$\alpha' = 2141$, $\beta' = 197913$ and $N = 734$, namely:
\begin{equation}
\label{eq:FastEAFDownMonth}
(5\cdot r_2 + 461)/153 = (2141\cdot r_2 + 197913)/2^{16}, \quad\forall r_2\in[0,
734[.
\end{equation}
\end{example}

In \cref{thm:FastEAFUp,thm:FastEAFDown}, $\alpha'$ and $\varepsilon$ are
obtained in $O(1)$ operations and $\beta'$ is found by an $O(\delta)$ search.
For each $r \in[0, \delta[$, $q(r)$ and $N(r)$ are calculated in $O(1)$
operations. $N$ and the overall search for a fast EAF then have $O(\delta)$
time complexity. Since many EAFs featuring in calendar calculations have small
divisors, finding more efficient alternatives for them is reasonably fast.

\subsection{Fast division -- the special case \texorpdfstring{{$\alpha = 1,
\beta = 0$}}{alpha = 1, beta = 0}}
\label{sec:SpecialCase}

This section examines $r/\delta$ with $\delta > 0$, {\em i.e.}, the
particular case of the EAF $f(r) = (\alpha\cdot r + \beta)/\delta$ where $\alpha
= 1$, $\beta = 0$ and $\delta > 0$.

For $\alpha = 1$, \cref{thm:FastEAFUp,thm:FastEAFDown} set $\alpha'$ to $\left
\lceil2^k\cdot\delta^{-1}\right\rceil$ and $\left\lfloor 2^k\cdot\delta^{-1}
\right\rfloor$, respectively. The former is the choice taken by
\cite{Alverson1991, CavagninoWerbrouck2007, GranlundMontgomery1994} and the
latter by \cite{MagenheimerPetersPettisZuras1987}. Finally, \cite{Robison2005}
and an appendix to \cite{CavagninoWerbrouck2007},
consider both and, in this particular case, rigourously justify the heuristics
we have suggested for choosing between the two approaches.

This section follows a more direct path but most of its results can be obtained
from the above for $\alpha = 1$ and $\beta = 0$. For instance, for $\alpha'>0$,
let $\beta' = - \min\{\alpha'\cdot r - 2^k\cdot f(r)\ ; \ r\in[0, \delta[\ \}$
as in \cref{thm:FastEAFUp}. Since $f(r) = r/\delta = 0$ for $r\in[0, \delta[$,
we have $\beta' = - \min\{\alpha'\cdot r\ ;\ r\in[0, \delta[\ \} = 0$. Hence, in
contrast to the general case, there is no need for an $O(\delta)$ search to obtain
$\beta'$. Similarly, $\beta' = \min\{2^k - 1 - (\alpha'\cdot r - 2^k\cdot f(r))\
; \ r\in[0, \delta[ \ \}$, as defined by \cref{thm:FastEAFDown}, can be proven to
be $\alpha' + (2^k\%\delta - 1)$, the same value found in
\cite{MagenheimerPetersPettisZuras1987}. \Cref{thm:FastEAFDown} assumes that $2^k\%
\delta\ge 1$ and, in the definition of $\beta'$, had we subtracted $2^k\%\delta$
instead of $1$ the proof would still work but $\beta'$ would have a smaller
value, namely, $\beta' = \alpha'$. In this case, the final reduction would be
$r/\delta = \alpha'\cdot(r + 1)/2^k$ as found in \cite{CavagninoWerbrouck2007,
Robison2005}. In addition, by making $\beta'$ smaller, $q(r)$, $M(r)$ and $N$
also decrease. Hence, \cref{thm:FastEAFDown} obtains a range of validity that is
no smaller than the one obtained in \cite{CavagninoWerbrouck2007, Robison2005}.

The remainder of this section focuses on the round-up approach but the
round-down alternative could be similarly considered.
Our reasons for this are that the round-up approach is mostly used by compilers
and that the appearance of the term $r + 1$ can potentially lead to an overflow.

\begin{lemma}
\label{lem:FastDivision}
Let $\delta,k\in\Zp$ with $\delta > 0$. Set $\alpha' = 2^k/\delta + 1$ and
$\varepsilon = \delta - 2^k\%\delta$. Then, $\alpha'\cdot\delta = 2^k +
\varepsilon$ and for any $n\in\Z$ we have:
\begin{equation}
\label{eq:FastDivision}
\alpha'\cdot n/2^k = n/\delta \quad\text{if, and only if,}\quad
0\le\varepsilon \cdot(n/\delta) + \alpha'\cdot(n\%\delta) < 2^k.
\end{equation}
\end{lemma}

\begin{proof}
Taking $\alpha = 1$ in \cref{thm:FastEAFUp} gives the same $\alpha'$ and
$\varepsilon$ as here. \Cref{eq:epsilonUp} then reads $\varepsilon = \alpha'
\cdot\delta - 2^k$. Similarly, from \cref{eq:FastEAFUp} with $f(n) = n/\delta$
and $f(n\%\delta) = n\%\delta/\delta = 0$ we obtain:
\begin{equation*}
\alpha'\cdot n - 2^k\cdot(n/\delta) = \varepsilon\cdot(n/\delta) + \alpha'\cdot
(n\%\delta).
\end{equation*}
The result follows from Euclidean division by $2^k$.
\end{proof}

Since $\alpha',\varepsilon\in\Zp$, we have $0\le\varepsilon\cdot(n/\delta) +
\alpha'\cdot(n\%\delta)$ for all $n\in\Zp$. There is an illuminating geometric
interpretation for the second part of the double
inequality in \cref{eq:FastDivision} that follows from mapping each $n\in\Zp$ to $P_n = (n/\delta, n
\%\delta)$ in the $xy$ plane. \cref{fig:Division} shows, for $\delta = 5$, some
points of particular interest labelled by their circled value. Notably, the
origin $P_0 = (0/5, 0\%5)$, $P_1 = (1/5, 1\%5)$ and $P_5 = (5/5, 5\%5)$.
Inequality $\varepsilon\cdot(n/\delta) + \alpha'\cdot(n\%\delta) < 2^k$ states
that $P_n$ is below the line $\varepsilon\cdot x + \alpha'\cdot y = 2^k$, which
is represented by the dotted line in \cref{fig:Division}, for $k = 8$, $\varepsilon = 4$
and $\alpha' = 52$. Hence, \cref{eq:FastDivision} means that a point below the
line corresponds to $n\in\Zp$ such that $n/\delta = \alpha'\cdot n/2^k$ and a
point above or on the line is related to $n$ such that $n/\delta\ne\alpha'\cdot
n/2^k$.
\begin{figure}[ht]
\hfill
\includegraphics[width=0.95\textwidth]{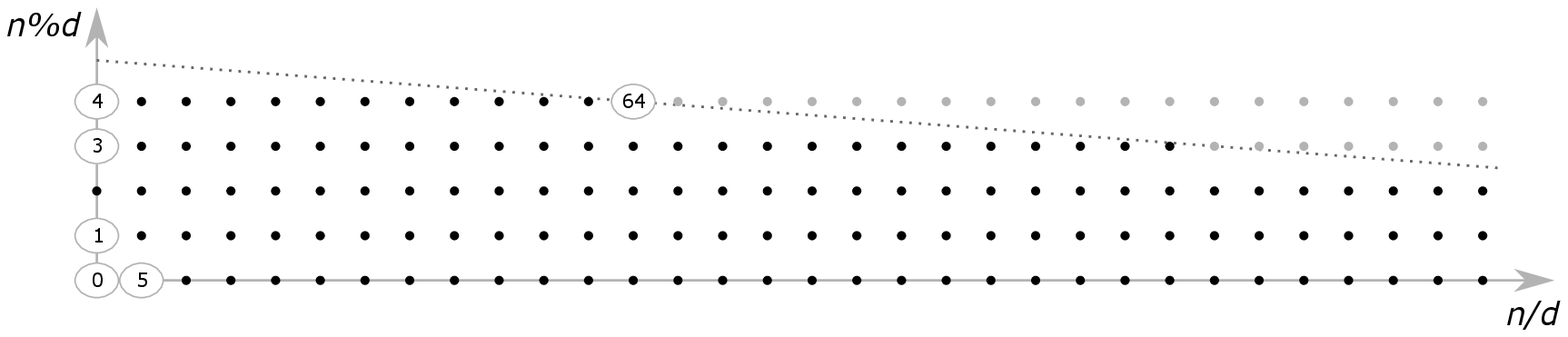}
\hfill{}
\caption{The geometry of replacing $n/\delta$ with $\alpha'\cdot n/2^k$, where
$\delta = 5$ and $k = 8$, $\alpha' = 2^k/\delta + 1 = 52$.}%
\label{fig:Division}
\end{figure}

If the slope of the dotted line is not too steep (more precisely, $-\varepsilon
\cdot\alpha'^{-1}\ge -1)$, then the graph makes it obvious that the smallest $N$ for
which $P_N$ is above or on the dotted line must lie on the line $y = \delta - 1
= 4$. In our case $N = 64$. For \cref{thm:FastEAFUp}, with $\alpha = 1$ and
$\beta = 0$, this means that $N = M(\delta - 1)$. In particular, we are not
interested in either $M(r)$ or $q(r)$ when $r\ne\delta - 1$. This and $\beta' =
0$ turn finding a fast EAF and its interval of applicability $[0, N[$ into an
$O(1)$ calculation rather than an $O(\delta)$ search as in the general case.
These geometric ideas are present although algebraically disguised in the proof
of \cref{thm:FastDivision}.

The result of \cref{thm:FastDivision} also appears in
\cite{CavagninoWerbrouck2007}. In addition to providing the aforementioned geometric insights
and an arguably simpler proof, we favour faster algorithms over range of
applicability. In other words, we reduce division to multiplication and bitwise shift
and find the largest $N$ for which this optimisation yields correct results for
all dividends in $[0, N[$. (In \cite{CavagninoWerbrouck2007} $N$ is called {\em
critical value} and is denoted by $N_{cr}$.)

\begin{theorem}
\label{thm:FastDivision}
Let $\delta,k\in\Zp$ with $\delta > 0$. Set $\alpha' = 2^k/\delta + 1$,
$\varepsilon = \delta - 2^k\%\delta$ and $N = \left\lceil\alpha'\cdot
\varepsilon^{-1}\right\rceil\cdot\delta - 1$. If $\varepsilon\le\alpha'$, then:
\begin{equation*}
n/\delta = \alpha'\cdot n/2^k, \quad\forall n\in[0, N[.
\end{equation*}
\end{theorem}

\begin{proof}
Let $N' = N - \delta = \left(\left\lceil\alpha'\cdot\varepsilon^{-1}\right\rceil
- 2\right)\cdot\delta + \delta - 1$. We have $N'/\delta = \left\lceil\alpha'
\cdot\varepsilon^{-1}\right\rceil - 2$ and $N'\%\delta = \delta - 1$. Recall
from \cref{lem:FastDivision} that $\alpha'\cdot\delta - \varepsilon = 2^k$.
Therefore, using $\varepsilon\cdot\left\lceil\alpha'\cdot\varepsilon^{-1}\right
\rceil < \alpha' + \varepsilon$ we obtain:
\begin{equation}
\label{eq:N'}
\varepsilon\cdot(N'/\delta) + \alpha'\cdot(N'\%\delta) = \varepsilon\cdot\left
\lceil\alpha'\cdot\varepsilon^{-1}\right\rceil - 2 \cdot\varepsilon + \alpha'
\cdot\delta - \alpha' < 2^k.
\end{equation}
Let $n\in\Zp$ with $n < N$. We shall show that $0\le\varepsilon\cdot(n/\delta) +
\alpha'\cdot(n\%\delta) < 2^k$ and the result will follow from
\cref{lem:FastDivision}. Since $\varepsilon$, $\alpha'$ and $r$ are
non-negative, it is sufficient to show that $\varepsilon\cdot(n/\delta) + \alpha'
\cdot(n\%\delta) < 2^k$.

Assume first that $n\le N'$. As a result, $n/\delta\le N'/\delta$. Then use \cref{eq:N'} to
obtain:
\begin{equation*}
\varepsilon\cdot(n/\delta) + \alpha'\cdot(n\%\delta)\le\varepsilon\cdot(N'/
\delta) + \alpha'\cdot(\delta - 1) = \varepsilon\cdot(N'/\delta) + \alpha'\cdot
(N'\%\delta) < 2^k.
\end{equation*}

Now assume that $N' + 1\le n < N$. Then $\left(\left\lceil\alpha'\cdot\varepsilon^
{-1}\right\rceil - 1\right)\cdot\delta\le n < \left(\left\lceil\alpha'\cdot
\varepsilon^{-1}\right\rceil - 1\right)\cdot\delta + \delta - 1$. Hence, $n/
\delta = \left \lceil\alpha'\cdot\varepsilon^{-1}\right\rceil - 1 = N'/
\delta + 1 = N/\delta$. Since $n < N$, we must have $n\%\delta\le N\%\delta
- 1 = \delta - 2$. Therefore,
\begin{align*}
\varepsilon\cdot(n/\delta) + \alpha'\cdot(n\%\delta)
&\le \varepsilon\cdot(N'/\delta + 1) + \alpha'\cdot(\delta - 2)
 = \varepsilon\cdot(N'/\delta) + \alpha'\cdot(\delta - 1) + \varepsilon -
\alpha' \\
&\le \varepsilon\cdot(N'/\delta) + \alpha'\cdot(N'\%\delta).
\end{align*}
The result follows from \cref{eq:N'}.
\end{proof}

\begin{remark}
If $\delta$ is not a power of two and $2^k\ge\delta\cdot(\delta - 2)$, then
$\varepsilon\le\alpha'$. Indeed, under these hypotheses we have $2^k/\delta\ge
\delta - 2$ and $2^k\%\delta\ge 1$. It follows that $\delta - 2^k\%\delta\le
\delta - 1\le 2^k/\delta + 1$, {\em i.e.}, $\varepsilon\le\alpha'$.
\end{remark}

\begin{example}
\label{ex:FastDivisionYear}
Consider the division $n_2/1461$. For $k = 39$ the constants given by
\cref{thm:FastDivision} are $\alpha' = 376\s287\s347$, $\varepsilon = 79$ and $N
= 6\s958\s934\s390$. Since $\varepsilon\le\alpha'$, this theorem gives:
\begin{equation}
\label{eq:FastDivisionYear39}
n_2/1461 = 376\s287\s347\cdot n_2/2^{39}, \quad \forall n_2\in[0,
6\s958\s934\s390[.
\end{equation}
Following \cite{GranlundMontgomery1994}, major compilers replace the expression
on the left with the one on the right when $n_2$ is a 32-bit unsigned integer.
(It is worth mentioning that $k = 39$ is the smallest value for which $N\ge
2^{32}$.)

For $k = 32$, the constants given by \cref{thm:FastDivision} are $\alpha' = 2
\s939\s745$, $\varepsilon = 149$ and $N = 28\s825\s529$. Hence,
\begin{equation}
\label{eq:FastDivisionYear32}
n_2 / 1461 = 2\s939\s745\cdot n_2/2^{32}, \quad \forall n_2\in[0, 28\s825\s529[.
\end{equation}
Although the interval is smaller than the one for $k = 39$, it is large enough
for the Gregorian calendar algorithm that we present later on.
\cref{ex:FastResidualYear} will unveil the advantage of
\cref{eq:FastDivisionYear32} over \cref{eq:FastDivisionYear39}.
\end{example}

\subsection{Fast evaluation of residual functions}
\label{sec:FastResidual}

\Cref{thm:FastEAFUp,thm:FastEAFDown} suggest optimisations for EAFs,
generalising the current practice for divisions revisited by
\cref{thm:FastDivision}. Our next result extends the optimisation to residual
functions.

\begin{theorem}
\label{thm:FastResidual}
Let $f$ and $f'$ be EAFs given by $f(r) = (\alpha\cdot r + \beta)/\delta$ and
$f'(r) = (\alpha'\cdot r + \beta')/\delta'$, with $\delta\ge\alpha > 0$, and
assume that $f\equiv f'$ on $[a, b[$. If $a = \hat{f}(f(a))$ and $f'(a - 1) <
f'(a)$, then:
\begin{equation*}
\mathring{f}(r)/\alpha = \mathring{f'}(r)/\alpha' \quad \forall r\in[a, b[.
\end{equation*}
\end{theorem}

\begin{proof}
Let $r\in[a, b[$ and set $q = f(r)$. From \cref{thm:EAF}-\labelcref{it:Residual},
we obtain $\mathring{f}(r)/\alpha = r - \hat{f}(q)$. We will finish the proof
by using \cref{lem:Residual} to show that $\mathring{f'}(r)/\alpha' = r -
\hat{f}(q)$.

Since $\delta\ge\alpha > 0$, both $f$ and $\hat{f}$ are non-decreasing. Hence,
from $a\le r\le b - 1$ we obtain $\hat{f}(f(a))\le\hat{f}(f(r))\le\hat{f}(f(b -
1))$. By assumption, $a = \hat{f}(f(a))$ and
\cref{thm:EAF}-\labelcref{it:Residual} (for $r = b - 1$) gives $b - 1\in[\hat{f}
(f(b - 1)), \hat{f}(f(b - 1) + 1)[$, in particular $\hat{f}(f(b - 1))\le b -
1$. Therefore, $a\le\hat{f}(f(r))\le b - 1$, in other words, $\hat{f}(q)\in[a, b[$.

Since $f\equiv f'$ in $[a, b[$ and $\hat{f}(q)$ is in this interval, we have
$f'(\hat{f}(q)) = f(\hat{f}(q))$. \Cref{thm:EAF}-\labelcref{it:RInverse} yields
$f(\hat{f}(q)) = q$ and thus $f'(\hat{f}(q)) = q = f(r)$.

Similarly, it can be shown that if $\hat{f}(q) - 1\in[a, b[$, then $f'(\hat{f}(q) -
1) = q - 1$, specifically, $f'(\hat{f}(q) - 1) < f'(\hat{f}(q))$. On the other
hand, if $\hat{f}(q) - 1\notin[a, b[$, then we must have $\hat{f}(q) = a$ and
the assumption on $f'$ also gives $f'(\hat{f}(q)  - 1) < f'(\hat{f}(q))$.

We have shown that $f'(\hat{f}(q) - 1) < f'(\hat{f}(q))$ and $f(r) = f'(\hat{f}
(q))$. Then, \cref{lem:Residual} applied to $f'$ (instead of $f$) and $g =
\hat{f}$ gives $\mathring{f'}(r)/\alpha' = r - \hat{f}(q)$.
\end{proof}

\begin{example}
\label{ex:FastResidualYear}
Consider the EAFs $f(n_2) = n_2/1461$, $\hat{f}(q) = q\cdot 1461$ and $f'(n_2) =
2\s939\s745\cdot n_2/2^{32}$. \Cref{eq:FastDivisionYear32} shows that $f(n_2) =
f'(n_2)$ for all $n_2\in [a, b[\ = [0, 28\s825\s529[$. Simple calculations give
$\hat{f}(f(0)) = 0$ and $f'(-1) = -1 < 0 = f'(0)$. Hence, it follows from
\cref{eq:FastDivisionYear32,thm:FastResidual} that:
\begin{equation*}
\label{eq:FastResidualYear}
n_2 / 1461 = 2\s939\s745\cdot n_2/2^{32} \quad\text{and}\quad
n_2 \% 1461 = 2\s939\s745\cdot n_2\%2^{32}/2\s939\s745, \quad \forall n_2\in[0,
28\s825\s529[.
\end{equation*}

Revisiting \cref{ex:ResidualYear} we show three alternatives that can be used to obtain $q_2 = (4
\cdot r_1 + 3)/1461$ and $r_2 = r_1 - 1461\cdot q_2/4$ side-by-side. The first
evaluates $r_2$ as expressed. The second uses the expression $(4\cdot r_1 + 3)\%
1461/4$ seen in \cref{ex:ResidualYear}. The third is similar but replaces $n_2/
1461$ and $n_2\%1461$ with the aforementioned alternatives.
\begin{equation*}
\begin{aligned}
n_2 &= 4\cdot r_1 + 3, \\
\\
q_2 &= n_2/1461, \\
r_2 &= r_1 - 1461\cdot q_2/4, \\
\end{aligned}
\quad\vrule\quad
\begin{aligned}
n_2 &= 4\cdot r_1 + 3, \\
\\
q_2 &= n_2/1461, \\
r_2 &= (n_2 - 1461\cdot q_2)/4,\ (\text{\em i.e., } r_2 = n_2 \% 1461/4) \\
\end{aligned}
\quad\vrule\quad
\begin{aligned}
n_2 &= 4\cdot r_1 + 3, \\
u_2 &= 2\s939\s745\cdot n_2, \\
q_2 &= u_2/2^{32}, \\
r_2 &= u_2\%2^{32}/2\s939\s745/4. \\
\end{aligned}
\end{equation*}
We are interested in instructions emitted by compilers for the three listings
above but, for clarity, we did not perform typical compiler optimisations ({\em
e.g.}, substituting division by $4$ and $2^{32}$ with bitwise shifts). However,
in the middle listing we did substitute $n_2\%1461$ with $n_2 - 1461\cdot q_2$
for greater resemblance to emitted instructions and fairer comparison with
the first listing. Compilers reduce division by $1461$
to multiplication and bitwise shift, the same operations used to get $u_2$ and $q_2$ in
the third listing. Therefore, up to and including the
calculation of $q_2$, the instructions for the three columns are the same and
only the constants differ.

The calculation of $r_2$ in the first two listings requires three instructions:
one multiplication, one subtraction and one bitwise shift $(/4$). For the third
listing, there {\em appear to be} four: one bitwise \texttt{and} ($\%2^{32}$), one
multiplication and one bitwise shift (compilers' reduction of $/2\s939\s745$), and another bitwise shift
($/4$). In fact, these two bitwise shifts collapse into one. Hence, there are
also three instructions for the third listing. We conclude that the number of
instructions is exactly the same for the three listings, but this is only part
of the story. The breakthrough of the third listing is removing the dependency of
$r_2$ on $q_2$, seen in the first two listings, which forces the CPU to wait for the
calculation of $q_2$ to finish before starting that of $r_2$. In listing three,
once $u_2$ is obtained, the evaluations of $q_2$ and $r_2$ can start
concurrently. There is a small but worthwhile price to pay for this
parallelisation and even that can be avoided in some platforms as we shall see
now.

The calculations of $q_2$ and $r_2$ need $u_2$ to be stored in two different
registers, which implies a \texttt{mov} from the register on which $u_2$ was
originally obtained. This does not necessarily increase the number of
instructions because instead of leaving the compiler to reduce division by
$1461$ to multiplication and a bitwise shift of $39$ bits (see
\cref{eq:FastDivisionYear39}), we perform our own reduction (see
\cref{eq:FastDivisionYear32}) using $32$ bits. For backward compatibility,
x86-64 CPUs provide \texttt{mov} instructions that reset the upper $32$ bits of
$64$ registers, effectively performing a \texttt{mov} and a bitwise \texttt{and}
at the same time. Hence, the operation $\%2^{32}$ in the calculation of $r_2$
might come for free.
\end{example}

\begin{example}
\label{ex:FastResidualMonth}
Consider the EAFs $f(r_2) = (5\cdot r_2 + 461)/153$, $\hat{f}(m_0) = (153\cdot
m_0 - 457)/5$, and $f'(r_2) = (2141\cdot r_2 + 197913)/2^{16}$.
\Cref{eq:FastEAFDownMonth} states that $f(r_2) = f'(r_2)$ for all $r_2\in[a, b[\
= [0, 734[$. Simple calculations give $\hat{f}(f(0)) = 0$ and $f'(-1) = 2 < 3 =
f'(0)$. It follows from \cref{eq:FastEAFDownMonth,thm:FastResidual} that:
\begin{align*}
(5\cdot r_2 + 461)/153 &= (2141\cdot r_2 + 197913)/2^{16} \quad\text{and} \\
(5\cdot r_2 + 461)\%153/5 &= (2141\cdot r_2 + 197913)\%2^{16}/2141,
\quad\forall r_2\in[0, 734[.
\end{align*}
\Cref{ex:ResidualMonth} shows that the left side of the last equation matches
$d_0 = r_2 - (153\cdot m_0 - 457)/5$. Hence, as in \cref{ex:FastResidualYear},
we have three alternatives for evaluating this quantity. The analysis of
\cref{ex:FastResidualYear} also holds here and once $n_3 = 2141\cdot r_2 +
197913$ is obtained the calculations of $m_0 = n_3/2^{16}$ and $d_0 = n_3\%2^
{16}/2141$ can start concurrently. Furthermore, in some platforms, the operation
$\%2^{16}$ might come for free.
\end{example}

\begin{example}
\label{ex:FastTime}
In line with the previous examples, we have
\begin{align*}
n/3600 &= 1\s193\s047\cdot n/2^{32}, &
n\%3600 &= 1\s193\s047\cdot n\%2^{32}/1\s193\s047, & \forall n&\in[0, 2\s257\s199[;
\\
n/60 &= 71\s582\s789\cdot n/2^{32}, &
n\%60 &= 71\s582\s789\cdot n\%2^{32}/71\s582\s789, &
\forall n&\in[0, 97\s612\s919[;
\\
n/10 &= 429\s496\s730\cdot n/2^{32}, &
n\%10 &= 429\s496\s730\cdot n\%2^{32}/429\s496\s730, &
\forall n&\in[0, 1\s073\s741\s829[.
\end{align*}
Compilers create a data dependency when evaluating the left side of the equalities above
but not when the expressions on the right are used. The first two lines can be used
in conversions of the seconds elapsed since the start of the day, a quantity in $[0,
86400[$, to hours, minutes and seconds. The third line can be used in
conversions of non-negative integers up to $9$ digits into their decimal
representations.
\end{example}

\subsection{Quick remainder -- the special case \texorpdfstring{{$\alpha = 1$,
$\beta = 0$}}{alpha = 1, beta = 0}}

Again for this particular case, the EAF $f(n) = (\alpha\cdot n + \beta)/\delta$
simplifies to $f(n) = n/\delta$ and its residual function simplifies to $\mathring{f}(n) =
n\%\delta$. In this section we use \cref{thm:FastDivision} and
\cref{thm:FastResidual} to derive an efficient way to calculate remainders.

Formally, \cref{thm:FastDivision} states how to replace $n/\delta$ with $\alpha'
\cdot n/2^k$, where $\alpha'\approx 2^k\cdot\delta^{-1}$ and $k\in\Zp$.
\cref{thm:FastResidual} then gives the equality $n\%\delta = (\alpha'\cdot n\%2^k)/
\alpha'$ and \cref{thm:FastDivision}, again, suggests replacing division by
$\alpha'$ with multiplication by an approximation of $2^k\cdot\alpha'^{-1}$ and
division by $2^k$. It turns out that $\delta\approx 2^k \cdot\alpha'^{-1}$ is the
approximation we need and we obtain $n\%\delta = \delta\cdot(\alpha'\cdot n\%2^
k)/2^k$. This is the idea behind our next theorem.

\begin{theorem}
\label{thm:FastRemainder}
Let $\delta,k\in\Zp$ with $\delta > 0$. Set $\alpha' = 2^k/\delta + 1$,
$\varepsilon = \delta - 2^k\%\delta$ and $M = \left\lceil 2^k\cdot\varepsilon^
{-1}\right\rceil$. If $\varepsilon\le\alpha'$, then:
\begin{equation}
\label{eq:FastRemainder}
n\%\delta = \delta\cdot(\alpha'\cdot n\%2^k)/2^k, \quad\forall n\in[0, M[.
\end{equation}
\end{theorem}

\begin{proof}
Set $N = \left\lceil\alpha'\cdot\varepsilon^{-1}\right\rceil\cdot\delta - 1\ge
\alpha'\cdot\varepsilon^{-1}\cdot\delta - 1$. Recall from
\cref{lem:FastDivision} that $\alpha'\cdot\delta = 2^k + \varepsilon$ and thus $N
\ge(2^k + \varepsilon)\cdot\varepsilon^{-1} - 1 = 2^k\cdot\varepsilon^{-1}$. It
follows that $N\ge M$. \Cref{thm:FastDivision} gives:
\begin{equation}
\label{eq:div}
n/\delta = \alpha'\cdot n/2^k, \quad\forall n\in[0, N[.
\end{equation}

Hence, for the EAFs $f(n) = n/\delta$ and $f'(n) = \alpha'\cdot n/2^k$,
\cref{eq:div} states that $f\equiv f'$ on $[0, N[$. Simple calculations give $\hat{f}(f
(0)) = \delta\cdot f(0) = 0$ and $f'(-1) < f'(0)$.
\cref{thm:FastResidual} (for $\alpha = 1$ and $[a, b[\ = [0, N[$) therefore yields:
\begin{equation}
\label{eq:mod}
n\%\delta = (\alpha'\cdot n\% 2^k)/\alpha', \quad\forall n\in[0, N[.
\end{equation}

Let $n\in[0, M[$ and set $m = \varepsilon\cdot(n/\delta) + \alpha'\cdot(n\%
\delta)$. Since $n < N$, we have $n/\delta\le N/\delta = \left\lceil\alpha'\cdot
\varepsilon^{-1}\right\rceil - 1 < \alpha'\cdot\varepsilon^{-1}$. It follows
that $0\le\varepsilon\cdot(n/\delta)<\alpha'$. Hence, $m/\alpha' = n\%\delta$
and $m\%\alpha' = \varepsilon\cdot(n/\delta)$. Moreover, \cref{eq:div,%
lem:FastDivision} give $0\le m < 2^k$. Since $\alpha'\cdot n = \alpha'\cdot
(\delta\cdot(n/ \delta) + n\%\delta) = 2^k\cdot(n/\delta) + \varepsilon\cdot(n/
\delta) + \alpha'\cdot(n\%\delta) = 2^k\cdot(n/\delta) + m$, we conclude that
$\alpha'\cdot n\%2^k = m$.

Since $n < M$, we have $\varepsilon\cdot n < 2^k$, otherwise $n\ge 2^k\cdot
\varepsilon^{-1}$ and then, $n\ge M$. Therefore:
\begin{equation}
\label{eq:m}
0\le\varepsilon\cdot(m/\alpha') + \delta\cdot(m\%\alpha')
= \varepsilon\cdot(n\%\delta) + \delta\cdot\varepsilon\cdot(n/\delta)
= \varepsilon\cdot n
< 2^k.
\end{equation}

Using $\alpha'\cdot\delta = 2^k + \varepsilon$ again gives $2^k = \delta\cdot
\alpha' - \varepsilon = (\delta - 1)\cdot\alpha' + \alpha' - \varepsilon$ and,
since $0\le\alpha' - \varepsilon < \alpha'$, we
obtain $2^k/\alpha' = \delta - 1$ and $2^k\%\alpha' = \alpha' - \varepsilon$,
that is, $\delta = 2^k/\alpha' + 1$ and $\varepsilon = \alpha' - 2^k\%\alpha'$.
Therefore, \cref{eq:m} and \cref{lem:FastDivision} (with $\alpha'$ and $\delta$
interchanged) give $m/\alpha' = \delta\cdot m/2^k$, namely, $(\alpha'\cdot n\%
2^k)/\alpha' = \delta\cdot(\alpha'\cdot n\% 2^k)/2^k$.

Finally, \cref{eq:mod} yields $n\%\delta = (\alpha'\cdot n\% 2^k)/\alpha' =
\delta\cdot(\alpha'\cdot n\% 2^k)/2^k$.
\end{proof}

\begin{example}
Revisiting \cref{ex:FastTime}, \cref{thm:FastRemainder} gives
\begin{align*}
n\%3600 &= 3600\cdot(1\s193\s047\cdot n\%2^{32})/2^{32},
& \forall n&\in[0, 2\s255\s761[;
\\
n\%60 &= 60\cdot(71\s582\s789\cdot n\%2^{32})/2^{32}, &
\forall n&\in[0, 97\s612\s894[;
\\
n\%10 &= 10\cdot(429\s496\s730\cdot n\%2^{32})/2^{32}, &
\forall n&\in[0, 1\s073\s741\s824[.
\end{align*}
The expressions on the right side of the equals sign provide efficient ways of
evaluating remainders. However, greater benefits are achieved when they are used
in conjunction with the quotient expressions presented in \cref{ex:FastTime}.
\end{example}

The equality in \cref{eq:FastRemainder} also appears in
\cite{LemireKaserKurz2019}. As in other works, it focuses on obtaining the value $k\in
\Zp$ for which the equality holds on an interval of the form $[0, 2^w[$ or, in
other words, $2^w\le M$ as the next result shows.

\begin{corollary}
Let $\delta, l, w\in\Zp$ with $0 < \delta < 2^w$ and $\delta - 2^{w + l}\%\delta
\le 2^l$. Set $\alpha' = 2^{w + l}/\delta + 1$. Then,
\begin{equation*}
n\%\delta = \delta\cdot(\alpha'\cdot n\%2^{w + l})/2^{w + l}, \quad\forall
n\in[0, 2^w[.
\end{equation*}
\end{corollary}

\begin{proof}
Set $k = w + l$, $\varepsilon = \delta - 2^k\%\delta$ and $M = \left\lceil 2^k
\cdot\varepsilon^{-1}\right\rceil$. From \cref{thm:FastRemainder}, it is sufficient to
show that $\varepsilon\le\alpha'$ and $2^k\le M$. Since $\delta < 2^w$ and
$\varepsilon\le 2^l$, we have $\delta\cdot\varepsilon < 2^{w + l} = 2^k$ and
thus, $\varepsilon\le 2^k/\delta < \alpha'$. We also have $M\ge 2^k\cdot
\varepsilon^{-1}\ge 2^k\cdot 2^{-l} = 2^w$.
\end{proof}

\section{Calendars}
\label{sec:Calendars}

We begin this section with a mathematical framework that can be applied to all calendars, before specifically turning our attention to the Gregorian calendar.

\begin{definition}
A {\bf date} is an ordered pair $(y, x)\in\Z\times X$ and a {\bf calendar} is a
non-empty set of dates $C\subset\Z\times X$ such that for all $y\in\Z$ the set
$\{ x\in X\ ;\  (y, x)\in C\}$ is finite.
\end{definition}

\begin{example}
\label{ex:ymd}
In the most common interpretation, $y$ is called {\bf year} and $x$ is the {\bf
day of the year}. Moreover, $X=\Z^2$ and $x = (m, d)$ is broken down into
sub-coordinates, namely {\bf month} $m$ and {\bf day (of the month)} $d$.
\end{example}

\begin{example}
\label{ex:cymd}
In another interpretation the coordinates of a date $(c, x)\in\Z\times\Z^3$ are
called {\bf century} $c$ and {\bf day of the century} $x$ and the sub-coordinates of
$x = (y, m, d)$ are the {\bf year of the century} $y$, {\bf month} $m$ and {\bf day}
$d$.
\end{example}

We mostly use the interpretation of \cref{ex:ymd}, although that of \cref{ex:cymd}
appears, in passing, in \cref{prp:Century,prp:Year}. To ease notation, a date
$(y, (m, d))\in\Z\times \Z^2$ is identified with $(y, m, d)\in\Z^3$.

Under the lexicographical order of $\Z^3$ (considered throughout), any calendar
$C\subset\Z^3$ is totally ordered and any bounded interval in $C$ is finite.
This observation supports the following definition.

\begin{definition}
Let $C\subset\Z^3$ be a calendar and $e\in C$. The function $\rho:C\rightarrow
\Z$ given by $\rho(x) = \#[e, x[$, if $x\ge e$, and $\rho(x) = -\#[x,
e[$, if $x < e$, is called the {\bf rata die} function with {\bf epoch} $e$.
\end{definition}

The term rata die is usually \cite{ReingoldDershowitz2018} applied to the
particular case where the epoch is 31 December 0000 in the proleptic Gregorian
calendar but we shall use it regardless of the epoch.

Rata die functions are strictly increasing and thus invertible on their images.
We are interested in developing algorithms to evaluate rata die functions and their
inverses.

\subsection{The Gregorian calendar}
\label{sec:Gregorian}

This calendar is used by most countries and is familiar to most readers.
Nevertheless, we recall some of its properties. Although it is meaningless to
refer to dates in the Gregorian calendar prior to its introduction in $1582$, we
can extrapolate the calendar backwards indefinitely, yielding the so called {\bf
proleptic Gregorian calendar}. (See Richards \cite{Richards1998}.) For the sake
of brevity, we drop the adjective {\em proleptic} and refer to the
Gregorian calendar even when dealing with dates prior to $1582$.

A year $y\in\Z$ is said to be a {\bf leap year} if $y\%4 = 0$ and $y\%100
\ne 0$ or $y\%400 = 0$. The Gregorian calendar is defined by $G = \{ (y, m, d)
\in\Z^3\ ;\ m\in[1, 12]\text{ and } d\in[1, L(y, m)] \}$, where $L(y, m)$ is the
{\bf length of month} $m$ of year $y$ as given by \cref{tab:Months}.
\begin{table}[ht]
\caption{Months of the Gregorian calendar.}
\label{tab:Months}
\hfill
\begin{tabular}{clc}
\hline\noalign{\smallskip}
$m$ & Name & $L(y, m)$ \\
\hline\noalign{\smallskip}
1 & January  & 31 \\
2 & February & 28 or 29 {\scriptsize $^{(a)}$} \\
3 & March    & 31 \\
4 & April    & 30 \\
\hline
\end{tabular}
\hfill
\begin{tabular}{clc}
\hline\noalign{\smallskip}
$m$ & Name & $L(y, m)$ \\
\hline\noalign{\smallskip}
5 & May    & 31 \\
6 & June   & 30 \\
7 & July   & 31 \\
8 & August & 31 \\
\hline
\end{tabular}
\hfill
\begin{tabular}{clc}
\hline\noalign{\smallskip}
$m$ & Name & $L(y, m)$ \\
\hline\noalign{\smallskip}
 9 & September & 30 \\
10 & October   & 31 \\
11 & November  & 30 \\
12 & December  & 31 \\
\hline
\end{tabular}
\hfill{}

\centerline{
\footnotesize {\scriptsize ${(a)}$} $L(y, 2) = 29$ if $y$ is a leap year, or
$28$ otherwise.
}
\end{table}

From the definition of length of month we obtain:
\begin{equation}
\label{eq:YearLength}
\sum_{m = 1}^{12} L(y, m) =
\begin{cases}
366, & \text {if $y$ is a leap year}, \\
365, & \text {otherwise}.
\end{cases}
\end{equation}

\subsection{The computational calendar}
\label{sec:Computational}

Borrowing Hatcher's \cite{Hatcher1984, Hatcher1985} terminology, we
introduce a computational calendar that allows efficient evaluations of rata
die functions and their inverses. This calendar is derived from $G$ by rotating
the months to place February last and setting its minimum to $e_0 = (0, 3,
0)$.

Let $P_1:\Z^3\rightarrow\Z^3$ be the map defined by:
\begin{equation*}
P_1(y_1, m_1, d_1) = (y_1 - {\bf 1}_{\{m_1\le 2\}},
m_1 + 12\cdot{\bf 1}_{\{m_1\le 2\}}, d_1 - 1),
\quad\forall (y_1, m_1, d_1)\in\Z^3.
\end{equation*}

\begin{remark}
\label{rmk:P1}
It is easy to see that $P_1$ is a strictly increasing bijection with
\begin{equation}
\label{eq:P1Inverse}
P_1^{-1}(y_0, m_0, d_0) = (y_0 + {\bf 1}_{\{m_0\ge 13\}},
m_0 - 12\cdot{\bf 1}_{\{m_0\ge 13\}}, d_0 + 1),
\quad\forall (y_0, m_0, d_0)\in G_0.
\end{equation}
\end{remark}

The {\bf computational calendar} is $G_0 = \{P_1(x_1)\ ;\ x_1\in G\text{ and }
P_1(x_1)\ge e_0 \}$. By setting $e_1 = (0, 3, 1)\in G$ and noting that
$P_1(e_1) = e_0$, the monotonicity of $P_1$ gives $G_0 = P(G_1)$, where $G_1 =
\{ x_1\in G\ ;\ x_1 \ge e_1\}$.

\Cref{fig:ComputationalCalendar} illustrates the relationships between years
and months in the Gregorian and computational calendars. Month $m_0 = 14$ of
year $0$ (generally, $y_0$) of the computational calendar corresponds
to $m_1 = 2$ (February) of year $1$ (generally, $y_0 + 1$) of the Gregorian
calendar. Hence, if $y_1 = y_0 + 1$ is a leap year in the Gregorian calendar,
then from \cref{eq:YearLength}, year $y_0$ of the computational calendar will have
$366$ days; otherwise it will have $365$ days.
\begin{figure}[ht]
\hfill
\includegraphics[width=0.8\textwidth]{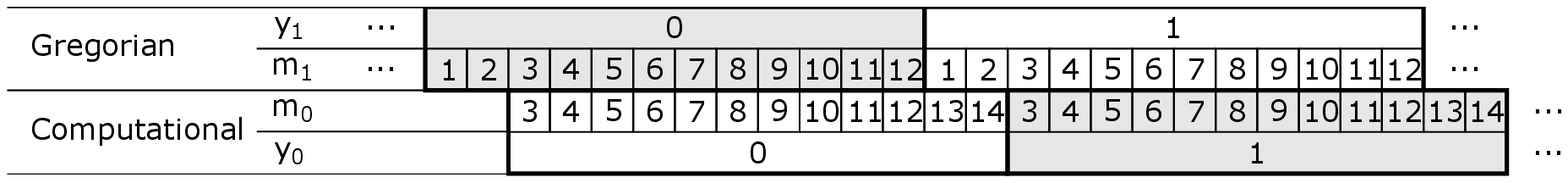}
\hfill{}
\caption{Years $0$ and $1$ in the Gregorian and computational calendars.}%
\label{fig:ComputationalCalendar}
\end{figure}

If $(y_0, m_0, d_0)\in G_0$, then $m_0\in[3, 15[$ and $d_0\in[0, L(y_1, m_1)[$
where $(y_1, m_1, d_1) = P_1^{-1}(y_0, m_0, d_0)$. (See
\cref{tab:ComputationalCalendar}.)
\begin{table}[ht]
\caption{Months of the computational calendar and their lengths.}
\label{tab:ComputationalCalendar}
\hfill
\begin{tabular}{clc}
\hline\noalign{\smallskip}
$m_0$ & Name & $L(y_1, m_1)${\scriptsize $^{(a)}$}  \\
\hline\noalign{\smallskip}
 3 & March & 31 \\
 4 & April & 30 \\
 5 & May   & 31 \\
 6 & June  & 30 \\
 7 & July  & 31 \\
\hline
\end{tabular}
\hfill
\begin{tabular}{clc}
\hline\noalign{\smallskip}
$m_0$ & Name & $L(y_1, m_1)${\scriptsize $^{(a)}$}  \\
\hline\noalign{\smallskip}
 8 & August    & 31 \\
 9 & September & 30 \\
10 & October   & 31 \\
11 & November  & 30 \\
12 & December  & 31 \\
\hline
\end{tabular}
\hfill
\begin{tabular}{clc}
\hline\noalign{\smallskip}
$m_0$ & Name & $L(y_1, m_1)${\scriptsize $^{(b)}$}  \\
\hline\noalign{\smallskip}
13 & January  & 31 \\
14 & February & 28 or 29 {\scriptsize $^{(c)}$} \\
\\
\\
\\
\hline
\end{tabular}
\hfill{}

{
\footnotesize
{\scriptsize ${(a)}$} $y_1 = y_0$, $m_1 = m_0$.
{\scriptsize ${(b)}$} $y_1 = y_0 + 1$, $m_1 = m_0 - 12$.
{\scriptsize ${(c)}$} $L(y_1, m_1) = 29$ if $y_1 = y_0 + 1$ is a leap year, or
$L(y_1, m_1) = 28$ otherwise.
}
\end{table}

\begin{remark}
\label{rmk:MonthPeriod}
\Cref{tab:ComputationalCalendar} shows that, except for $m_0 = 14$, $m_0\mapsto
L(y_1, m_1)$ is a periodic function and there are $153 = 31 + 30 + 31 + 30 + 31$
days in each $5$-month period.
\end{remark}

\section{A rata die function on the computational calendar}
\label{sec:RataDie}

$e_0$, the minimum date in $G_0$, is undoubtedly the most natural epoch for a
rata die function $\rho_0$ on $G_0$, $\rho_0(x_0) = \#[e_0, x_0[$ for all $x_0
\in G_0$.

For any $x_0 = (y_0, m_0, d_0)\in G_0$, integers $y_0$, $m_0$, $d_0$ and
$\rho_0(x_0)$ are non-negative. Hence, implementations of $G_0$ can work
exclusively on unsigned integer types, which are usually faster than signed
ones.

Let $x_0 = (y_0, m_0, d_0)\in G_0$ and split $[e_0, x_0[$ into three disjoint
intervals:
\begin{equation*}
[e_0, x_0[\ = [(0, 3, 0), (y_0, 3, 0)[\ \cup\ [(y_0, 3, 0), (y_0, m_0, 0)[\ \cup
\ [(y_0, m_0, 0), (y_0, m_0, d_0)[.
\end{equation*}
Hence, $\rho_0(x_0)$ is the sum of the number of elements of these three
intervals, which we consider separately.

The {\bf year count}, defined by $y_c(y_0) = \#[(0, 3, 0), (y_0, 3, 0)[$ is the
number of dates prior to year $y_0$:
\begin{equation*}
y_c(y_0)
= 365\cdot y_0 + \#\{ y_0'\in [0, y_0[\ ;\ y_0' \text{ has 366 days }\} 
\end{equation*}
Recalling that year $y_0'$ of the computational calendar has 366 days if, and
only if, $y_0' + 1$ is a leap year of the Gregorian calendar gives:
\begin{equation}
\label{eq:y0}
\begin{aligned}
y_c(y_0)
&= 365\cdot y_0 + \#\{ y_0'\in [0, y_0[\ ;\ y_0' + 1\text{ is a leap year }\} \\
&= 365\cdot y_0 + \#\{ y\in [1, y_0 + 1[\ ;\ y\text{ is a leap year }\} \\
&= 365\cdot y_0 + \#\{ y\in [1, y_0]\ ;\ y\%4 = 0 \text{ and }
y\%100\ne 0 \text{ or } y\% 400 = 0\} \\
&= 365\cdot y_0 + y_0/4 - y_0/100 + y_0/400.
\end{aligned}
\end{equation}

\begin{remark}
\label{rmk:LeapCycles}
For all $y_0$, $z\in\Z$ with $y_0 + 400\cdot z\ge 0$, we have
\begin{align*}
y_c(y_0 + 400\cdot z) &=
365\cdot (y_0 + 400\cdot z) + (y_0 + 400\cdot z)/4 - (y_0 + 400\cdot z)/100 +
(y_0 + 400\cdot z)/400 \\
&= 365\cdot y_0 + y_0/4 - y_0/100 + y_0/400 + 146000\cdot z + 100\cdot z - 4
\cdot z + z
\\
&= y_c(y_0) + 146097\cdot z.
\end{align*}
Hence, by adding $400\cdot z$ years to a date, its rata die increases by
$146097\cdot z$. There are $146097$ dates in any $400$-year interval. Such
intervals are called {\bf leap cycles}.
\end{remark}

The {\bf month count} defined by $m_c(y_0, m_0) = \#[(y_0, 3, 0), (y_0, m_0,
0)[$ is the number of days between the 1st day of year $y_0$ and the 1st day of
month $m_0$ of year $y_0$. It can be obtained by adding the lengths of previous
months in year $y_0$. Since $m_0 = 14$ is the last month of the year, its length
is not included in the summation and becomes irrelevant to $m_c$.
Therefore, the year $y_0$ also becomes irrelevant to $m_c$,
because the last month is the only one whose length depends on the year.
In other words, $m_c(y_0, m_0) = m_c(m_0)$ as per \cref{tab:MonthCount}.
\begin{table}[ht]
\caption{Month count.}
\label{tab:MonthCount}
\hfill
\begin{tabular}{ccc}
\hline\noalign{\smallskip}
$m_0$ & $L(y_1, m_1)$ & $m_c(m_0)$ \\
\hline\noalign{\smallskip}
 3 & 31 &   0 \\
 4 & 30 &  31 \\
 5 & 31 &  61 \\
 6 & 30 &  92 \\
 7 & 31 & 122 \\
\hline
\end{tabular}
\hfill
\begin{tabular}{ccc}
\hline\noalign{\smallskip}
$m_0$ & $L(y_1, m_1)$ & $m_c(m_0)$ \\
\hline\noalign{\smallskip}
 8 & 31 & 153 \\
 9 & 30 & 184 \\
10 & 31 & 214 \\
11 & 30 & 245 \\
12 & 31 & 275 \\
\hline
\end{tabular}
\hfill
\begin{tabular}{ccc}
\hline\noalign{\smallskip}
$m_0$ & $L(y_1, m_1)$ & $m_c(m_0)$ \\
\hline\noalign{\smallskip}
13 & 31           & 306 \\
14 & (irrelevant) & 337 \\
\\
\\
\\
\hline
\end{tabular}
\hfill{}
\end{table}

Some implementations \cite{LinuxTimeconv, GlibcOfftime, DotNetDateTime} use
look-up arrays to recover $m_c(m_0)$, but it can be expressed by an EAF:
\begin{equation}
\label{eq:m0}
m_c(m_0) = (153\cdot m_0 - 457)/5, \quad\forall m_0\in [3, 14].
\end{equation}
Variations of this formula have been found by many authors \cite{Baum1998,
Hatcher1984, Hatcher1985, BoostGregorianCalendar, LLVMChrono, DotNetDateTime,
ReingoldDershowitz2018, Richards1998}, most often through a trial-and-error
line fitting to the set of points $\big\{(m_0, m_c(m_0))\in\Z^2\ ;\ m_0\in [3,
14]\big\}$. \Cref{rmk:MonthPeriod} clarifies the matter.
However, a simple textbook linear regression on the set $\big\{(m_0, m_c(m_0) +
0.5)\in\Q^2\ ;\ m_0\in[3, 14]\big\}$ finds $m_c(m_0) = (26256\cdot m_0
- 78317)/858$. Using this EAF in \cref{thm:FastEAFDown} produces the same
faster alternative seen in \cref{ex:FastEAFDownMonthCount}.

The {\bf day count} defined by $d_c(y_0, m_0, d_0) = \#[(y_0, m_0, 0), (y_0,
m_0, d_0)[$ is the number of dates since the 1st of month $m_0$ of year $y_0$.
Trivially, $d_c(y_0, m_0, d_0) = d_0$.

Putting together the above results yields the next proposition.

\begin{proposition}
\label{prp:RataDie}
Let $(y_0, m_0, d_0)\in G_0$. Set:
\begin{align*}
y_c &= 365\cdot y_0 + y_0/4 - y_0/100 + y_0/400, &
m_c &= (153\cdot m_0 - 457)/5, &
d_c &= d_0.
\end{align*}
Then, $\rho_0(y_0, m_0, d_0) = y_c + m_c + d_c$.
\end{proposition}

The calculations of \cref{prp:RataDie} can be performed more efficiently by first computing the century.

\begin{corollary}
\label{cor:RataDie}
Let $(y_0, m_0, d_0)\in G_0$. Set:
\begin{align*}
q_1 &= y_0/100, &
y_c &= 1461\cdot y_0/4 - q_1 + q_1/4, &
m_c &= (979\cdot m_0 - 2919)/2^5, &
d_c &= d_0.
\end{align*}
Then, $\rho_0(y_0, m_0, d_0) = y_c + m_c + d_c$.
\end{corollary}

\begin{proof}
It suffices to show that the following alternatives for the expressions seen in
\cref{prp:RataDie} hold.
\begin{enumerate}
\item\label{it:365} $1461\cdot y_0/4 = 365\cdot y_0 + y_0/4$:
This follows from $1461\cdot y_0 = 4\cdot (365 \cdot y_0 + y_0/4) + y_0\%4$ and
$0\le y_0\%4 < 4$.

\item $(153\cdot m_0 -457)/5 = (979\cdot m_0 - 2919)/2^5$:
This follows from $m_0\in[3, 14]$ and \cref{eq:FastEAFDownMonthCount}.

\item $y_0/400 = q_1/4$:
This follows from $y_0/100/4 = y_0/400$.
\end{enumerate}
\end{proof}

\subsection{Inverting the rata die function on the computational calendar}
\label{sec:RataDieInv}

The computational calendar $G_0$ is unbounded from above and, thus, for any $r_0
\in\Zp$ a unique $x_0\in G_0$ exists such that $\rho_0(x_0) = r_0$. The objective of
this section is to find $x_0$ given that $r_0$.

For $x_0 = (y_0, m_0, d_0)\in G_0$, the quotient $q = y_0/100$ is its {\bf
century}. Since $100\cdot q\leq y_0 < 100\cdot (q + 1)$, we have $(100\cdot q,
3, 0)\leq x_0 < (100 \cdot (q + 1), 3, 0)$. The quantity $\#[(100\cdot q, 3, 0),
x_0[$ is the {\bf day of the century} of $x_0$. The next proposition retrieves the
century and the day of the century from a rata die.

\begin{proposition}
\label{prp:Century}
Let $x_0\in G_0$ and $r_0$ be its rata die. Set:
\begin{align*}
n_1 &= 4\cdot r_0 + 3, & q_1 &= n_1/146097, & r_1 &= n_1\%146097/4.
\end{align*}
Then, $q_1$ is the century of $x_0$ and $r_1$ is its day of the century.
\end{proposition}

\begin{proof}
Write $x_0 = (y_0, m_0, d_0)$ and set $q = y_0/100$ and $r = \#[(100\cdot q, 3,
0), x_0[$. We must prove that $q_1 = q$ and $r_1 = r$.

For $p\in\Zp$, let $g(p) = \#[e_0, (100\cdot p, 3, 0)[$. The definition of $y_c$
gives $g(p) = y_c(100\cdot p)$ and \cref{eq:y0} yields
\begin{align*}
g(p) &=
365\cdot 100\cdot p + 100\cdot p/4 - 100\cdot p/100 + 100\cdot p/400
 = 36524\cdot p + p/4 \\
&= 146097\cdot p/4.
\end{align*}
(The last equality is obtained as \cref{it:365} in the proof of \cref{cor:RataDie},
using $36524$ instead of $365$.)

Since $r_0 = \#[e_0, x_0[\ = \#[e_0, (100\cdot q, 3, 0)[\ +\ \#[(100\cdot q, 3,
0), x_0[$, we have $r_0 = g(q) + r\ge g(q)$. Now, $x_0 < (100\cdot(q
+ 1), 3, 0)$ and $r_0 = \#[e_0, x_0[\ < \#[e_0, (100\cdot (q + 1), 3, 0)[\ =
g(q + 1)$. Hence, $r_0 < g(q + 1)$, in other words, $r_0\in[g(q),
g(q + 1)[$. As seen in \cref{ex:ResidualCentury} (where $\hat{f} = g$):
\begin{equation*}
(4\cdot r_0 + 3)/146097 = q \quad\text{and}\quad (4\cdot r_0 + 3)\%146097/4 =
r_0 - g(q).
\end{equation*}
These equalities mean that $q_1 = q$ and $r_1 = r$.
\end{proof}

For $x_0 = (y_0, m_0, d_0)\in G_0$, the remainder $y_0\%100$ is its {\bf year of
the century}. We have $(y_0, 3, 0)\leq x_0 < (y_0 + 1, 3, 0)$ and the quantity
$\#[(y_0, 3, 0), x_0[$ is called the {\bf day of the year} of $x_0$. The next
proposition starts where \cref{prp:Century} left off, {\em i.e.}, at the day of
the century, to obtain the year of the century and the day of the year.

\begin{proposition}
\label{prp:Year}
Let $x_0\in G_0$ and $r_1$ be its day of the century. Set:
\begin{align*}
n_2 &= 4\cdot r_1 + 3, & q_2 &= n_2/1461, & r_2 &= n_2\%1461/4.
\end{align*}
Then, $q_2$ is the year of the century of $x_0$ and $r_2$ is its day of the year.
\end{proposition}

\begin{proof}
Write $x_0 = (y_0, m_0, d_0)$ and set $q_1 = y_0/100$, $q = y_0\%100$ and $r =
\#[(y_0, 3, 0), x_0[$. We must prove that $q_2 = q$ and that $r_2 = r$.

For $p\in[0, 100[$, it is easy to show that $(100\cdot q_1 + p)/100 = 100\cdot
q_1/100$ and $(100\cdot q_1 + p)/400 = q_1/4$. Let $g(p) = \#[(100\cdot q_1, 3,
0), (100\cdot q_1 + p, 3, 0)[$. The definition of $y_c$ gives $g(p) = y_c(100
\cdot q_1 + p) - y_c(100\cdot q_1)$ and \cref{eq:y0} yields:
\begin{align*}
g(p) &=
365\cdot(100\cdot q_1 + p) + (100\cdot q_1 + p)/4 - (100\cdot q_1 + p)/100 +
  (100\cdot q_1 + p)/400 - y_c(100\cdot q_1) \\
&= 365\cdot p + p/4
 = 1461\cdot p/4.
\end{align*}
(The last equality is the same as \cref{it:365} in the proof of \cref{cor:RataDie}.)

Since $y_0 = 100\cdot q_1 + q$ and $r_1 = \#[(100\cdot q_1, 3, 0), x_0[$, we
have
\begin{align*}
r_1
&= \#[(100\cdot q_1, 3, 0), (y_0, 3, 0)[\ +\ \#[(y_0, 3, 0), x_0[ \\
&= \#[(100\cdot q_1, 3, 0), (100\cdot q_1 + q, 3, 0)[\ +\ \#[(y_0, 3, 0), x_0[\
 = g(q) + r\ge g(q).
\end{align*}
Since $x_0 < (y_0 + 1, 3, 0) = (100\cdot q_1 + q + 1, 3, 0)$, we have $r_1 <
\#[(100\cdot q_1, 3, 0), (100\cdot q_1 + q + 1, 3, 0)[\ = g(q + 1)$ and it
follows that $r_1\in[g(q), g(q + 1)[$. As seen in \cref{ex:ResidualYear} (where
$\hat{f} = g$):
\begin{equation*}
(4\cdot r_1 + 3)/1461 = q \quad\text{and}\quad (4\cdot r_1 + 3)\%1461/4 = r_1
- g(q).
\end{equation*}
These equalities mean that $q_2 = q$ and that $r_2 = r$.
\end{proof}

Like \cref{prp:Year}, the following result provides calculations for year of century and day of year, but in a computationally more efficient way.

\begin{corollary}
\label{cor:Year}
Let $x_0\in G_0$ and $r_1$ be its day of the century. Set:
\begin{align*}
n_2 &= 4\cdot r_1 + 3, & u_2 &= 2\s939\s745\cdot n_2, & q_2 &= u_2/2^{32}, &
r_2 &= u_2\%2^{32}/2\s939\s745/4.
\end{align*}
Then, $q_2$ is the year of the century of $x_0$ and $r_2$ is its day of the year.
\end{corollary}

\begin{proof}
\Cref{ex:FastResidualYear} shows that the $q_2$ and $r_2$ set above match those of
\cref{prp:Year}, provided that $n_2\in[0, 28\s825\s529[$, which we shall
prove now. \Cref{prp:Century} provides that $r_1 = n_1\% 146097/4$ for some $n_1\in\Z$.
From $0 \le n_1\%146097 < 146097$, it easily follows that $3\le 4\cdot r_1 + 3 <
146100$, in other words, $n_2\in [3, 146100[$.
\end{proof}

The following proposition extracts the month and day from the day of the year.

\begin{proposition}
\label{prp:Month}
Let $x_0\in G_0$ and $r_2$ be its day of the year. Set:
\begin{align*}
n_3 &= 5\cdot r_2 + 461, & q_3 &= n_3/153, & r_3 &= n_3\%153/5.
\end{align*}
Then, $q_3$ is the month of $x_0$ and $r_3$ is its day.
\end{proposition}

\begin{proof}
Let $x_0 = (y_0, m_0, d_0)$. We must prove that $q_3 = m_0$ and that $r_3 = d_0$.

For $p\in[3, 15]$, let $g(p) = (153\cdot p - 457)/5$. \Cref{eq:m0} provides that if
$p\in[3, 14]$, then $g(p) = m_c(p) = \# [(y_0, 3, 0), (y_0, p, 0)[$. By definition,
$r_2 = \#[(y_0, 3, 0), x_0[$ can be written as:
\begin{align*}
r_2
&= \#[(y_0, 3, 0), (y_0, m_0, 0)[\ +\ \#[(y_0, m_0, 0), (y_0, m_0, d_0)[\
 = g(m_0) + d_0 \ge g(m_0).
\end{align*}
We shall now prove that $r_2 < g(m_0 + 1)$. If $m_0 < 14$, then $x_0 < (y, m_0 +
1, 0)$ and we obtain $r_2 = \#[(y_0, 3, 0), x_0[\ < \#[(y_0, 3, 0), (y_0, m_0 +
1, 0)[\ = m_c(m_0 + 1) = g(m_0 + 1)$. Now, if $m_0 = 14$, then $x_0\le (y_0, 14,
L(y_0, 2))$ and it follows that $r_2\leq\#[(y_0, 3, 0), (y_0, 14, L(y_0,
2))]\leq 366 < 367 = g(15) = g(m_0 + 1)$. Hence, $r_2 < g(m_0 + 1)$, in other words,
$r_2\in[g(m_0), g(m_0 + 1)[$. As seen in \cref{ex:ResidualMonth} (where $\hat{f}
= g$):
\begin{equation*}
(5\cdot r_2 + 461)/153 = m_0 \quad \text{and}\quad (5\cdot r_2 + 461)\%153/5 =
r_2 - g(m_0).
\end{equation*}
These equalities mean that $q_3 = m_0$ and that $r_3 = d_0$.
\end{proof}

\begin{remark}
\label{rmk:JanOrFeb}
In the proof of \cref{prp:Month}, we have shown that if $r_2$ is the day of the year
of $(y_0, m_0, d_0)\in G_0$ and $m_0 < 14$, then $r_2\in[m_c(m_0), m_c(m_0 +
1)[$. It follows that $m_0\ge 13$ is equivalent to $r_2\ge m_c(13) = 306$.
\end{remark}

Again, the following corollary improves the computational efficiency of the calculations of \cref{prp:Month}.

\begin{corollary}
\label{cor:Month}
Let $x_0\in G_0$ and $r_2$ be its day of the year. Set:
\begin{align*}
n_3 &= 2141\cdot r_2 + 197913, & q_3 &= n_3/2^{16}, & r_3 &= n_3\%2^{16}/2141.
\end{align*}
Then, $q_3$ is the month of $x_0$ and $r_3$ is its day.
\end{corollary}

\begin{proof}
\Cref{ex:FastResidualMonth} shows that the $q_3$ and $r_3$ set above match those of
\cref{prp:Month}, provided that $r_2\in[0, 734[$, which we shall prove now.
\Cref{prp:Year} provides that $r_2 = n_2\%1461/4$ for some $n_2\in\Z$. From $0\le
n_2\%1461 < 1461$, it easily follows that $0\le r_2\le 365$.
\end{proof}

\section{Rata die functions on the Gregorian calendar}
\label{sec:RataDieGregorian}

\Cref{sec:RataDie} covered evaluation of a rata die function and its inverse on
the computational calendar $G_0$. This section adapts them to the Gregorian
calendar.

$G_0$ allows implementations to work exclusively on unsigned integer types, which
might be faster than signed ones. To avoid compromising this advantage,
we only consider subsets of $G$ that have a minimum. These calendars, their
respective rata die functions and epochs are summarised in \cref{tab:Calendars}.
For ease of reference, the first row is a reminder for $G_0$.
\begin{table}[ht]
\caption{Gregorian-related calendars.}
\label{tab:Calendars}
\hfill
\begin{tabular}{lll}
\hline\noalign{\smallskip}
Calendar & Rata die & Epoch \\
\hline\noalign{\smallskip}
$G_0 = P_1(G_1)$, where $G_1$ is defined below.
  & $\rho_0:G_0\rightarrow\Z$
  & $e_0 = (0, 3, 0) = P_1(e_1)$
  \\
$G_1 = \{ x_1\in G\ ;\ x_1\ge e_1 \}$
  & $\rho_1:G_1\rightarrow\Z$
  & $e_1 = (0, 3, 1)$
  \\
$G_2 = \{ x_2\in G\ ;\ x_2\ge e_2 \}$
  & $\rho_2:G_2\rightarrow\Z$
  & $e_2 = (z_2, 3, 1)$, where $z_2$ is a multiple of $400$.
  \\
$G_2$ (as above)
  & $\rho_3:G_2\rightarrow\Z$
  & $e_3\in G_2$.
  \\
\noalign{\smallskip}\hline
\end{tabular}
\hfill{}
\end{table}

The rows of \cref{tab:Calendars} show increasing degrees of freedom of how epochs
and minima can be set. For $G_0$ and $\rho_0$ there are no choices: epoch and
the minimum date are set to $e_0 = (0, 3, 0)$. A match between the epoch and the minimum
implies that $\rho_0$ is non-negative. Since $e_0$'s year is zero, dates in
$G_0$ have non-negative years. The same holds for $G_1$, $\rho_1$ and $e_1 = (0,
3, 1)$. For $G_2$ and $\rho_2$, epoch and minimum also match and are set to $e_2
= (z_2, 3, 1)$, but $z_2$ can be any multiple of $400$. Therefore, $\rho_2$ is
also non-negative, although negative years are possible when $z_2 < 0$. Finally,
$\rho_3$ provides a more flexible setting where any epoch and minimum can be chosen.

As we shall see, $\rho_1 = \rho_0\circ P_1$. In general, maps from one calendar
to another provide a way to express rata die functions. Conversely, rata die
functions provide a way to map one calendar into another \cite{Hatcher1985,
ReingoldDershowitz2018, Richards1998}.

\begin{proposition}
\label{prp:Maps}
Let $C$ and $C'$ be calendars and let $\rho$ and $\rho'$ be their respective rata die
functions with epochs $e\in C$ and $e'\in C'$. If $P:C\rightarrow C'$ is a
strictly increasing bijection and $P(e) = e'$, then $\rho = \rho'\circ P$ and
$\rho^{-1} = P^{-1}\circ \rho'^{-1}$ on $Im(\rho') = Im(\rho)$. Conversely, if
$Im(\rho) = Im(\rho')$, then $\rho'^{-1}\circ \rho$ is a strictly increasing
bijection from $C$ to $C'$.
\end{proposition}

\begin{proof}
Let $x\in C$ and assume that $x\ge e$. Since $P$ is a strictly increasing bijection,
we have $\rho(x) = \#[e, x[\ = \#P([e, x[) = \#[P(e), P(x)[) = \#[e', P(x)[\ =
\rho'(P(x))$. The case $x < e$ is treated analogously and we conclude that $\rho
= \rho'\circ P$. From this and well-known results on functions and their
inverses, we obtain $\rho^{-1} = P^{-1}\circ \rho'^{-1}$ on $Im(\rho') =
Im(\rho)$.

By definition, $\rho(x)\in Im(\rho) = Im(\rho')$ and $\rho'^{-1}(\rho(x))\in
C'$, in other words, $\rho'^{-1}\circ\rho$ is a map from $C$ to $C'$.
Since $\rho:C\rightarrow\Z$ and $\rho'^{-1}:Im(\rho')\rightarrow C'$
are strictly increasing surjective functions, hence bijections,
their composition is again a strictly increasing bijection.
\end{proof}

\subsection{The non-negative rata die with epoch 1 March 0000}
\label{sec:RataDie031}

This section concerns the rata die function $\rho_1$ on $G_1 = \{x_1\in G\ ;\
x_1\ge e_1 \}$ with epoch $e_1 = (0, 3, 1)$.

From \cref{rmk:P1} and \cref{prp:Maps} applied to $C = G_1$, $C' = G_0$, $\rho =
\rho_1$, $\rho' = \rho_0$, $e = e_1$, $e' = e_0$ and $P = P_1$ we obtain:
\begin{equation}
\label{eq:r1r0}
(y_1, m_1, d_1)\in G_1 \quad\Longrightarrow\quad
P_1(y_1, m_1, \rho_1)\in G_0 \quad\text{and}\quad
\rho_1(y_1, m_1, d_1) = \rho_0\left(P_1(y_1, m_1, d_1)\right)
\end{equation}
and, conversely,
\begin{equation}
\label{eq:r0r1}
(y_0, m_0, d_0)\in G_0 \quad\Longrightarrow\quad
P_1^{-1}(y_0, m_0, \rho_0)\in G_1 \quad\text{and}\quad
\rho_0(y_0, m_0, d_0) = \rho_1\left(P_1^{-1}(y_0, m_0, d_0)\right).
\end{equation}

\subsection{Rata die with a chosen epoch}
\label{sec:RataDieE}

Let $z_2\in\Z$ be a multiple of $400$ and $e_2 = (z_2, 3, 1)$. This section
concerns two rata die functions: $\rho_2$ on $G_2 = \{ x_2\in G\ ; \ x_2\ge e_2
\}$ with epoch $e_2$, and $\rho_3$ also on $G_2$ but with epoch $e_3\in G_2$
arbitrarily chosen.

Let $P_2:\Z^3\rightarrow\Z^3$ be the map $P_2(y_2, m_2, d_2) = (y_2 - z_2, m_2,
d_2)$, which is a strictly increasing
bijection with $P_2^{-1}(y_1, m_1, d_1) = (y_1 + z_2, m_1, d_1)$ and $P_2(e_2) =
(0, 3, 1) = e_1$.

Let $(y_2, m_2, d_2)\in G_2$ and set $(y_1, m_1, d_1) = P_2(y_2, m_2, d_2)$. We
have $m_1 = m_2\in[1, 12]$, $d_1 = d_2\in[1, L(y_2, m_2)]$ and, since $z_2\%400
= 0$, $y_2 - z_2$ is a leap year if, and only if, $y_2$ is a leap year. Hence,
$L(y_2, m_2) = L(y_2 - z_2, m_2) = L(y_1, m_1)$ and it follows that $d_1\in[1,
L(y_1, m_1)]$. Therefore, $(y_1, m_1, d_1)\in G$. From $(y_2, m_2, d_2)\ge e_2$,
we obtain $(y_1, m_1, d_1)\ge e_1$ and conclude that $(y_1, m_1, d_1)\in G_1$.
Similarly, we can show that $P_2^{-1}(y_1, m_1, d_1)\in G_2$ for all $(y_1, m_1,
d_1)\in G_1$ and thus, $P_2$ is a bijection from $G_2$ to $G_1$.

Applying \cref{prp:Maps} to $C = G_2$, $C' = G_1$, $\rho = \rho_2$, $\rho' =
\rho_1$, $e = e_2$, $e' = e_1$ and $P = P_2$ gives:
\begin{equation}
\label{eq:r2r1}
(y_2, m_2, d_2)\in G_2 \quad\Longrightarrow\quad
(y_2 - z_2, m_2, d_2)\in G_1 \quad\text{and}\quad
\rho_2(y_2, m_2, d_2) = \rho_1(y_2 - z_2, m_2, d_2)
\end{equation}
and, reciprocally,
\begin{equation}
\label{eq:r1r2}
(y_1, m_1, d_1)\in G_1 \quad\Longrightarrow\quad
(y_1 + z_2, m_1, d_1)\in G_2 \quad\text{and}\quad
\rho_1(y_1, m_1, d_1) = \rho_2(y_1 + z_2, m_1, d_1).
\end{equation}

Fix $e_3\in G_2$ and let $\rho_3:G_2\rightarrow\Z$ be the rata die function with
epoch $e_3$. Let $x_2\in G_2$. On the one hand, if $x_2\ge e_3$, then $\#[e_2, x_2[
\ = \#[e_2, e_3[\ +\ \#[e_3, x_2[$, in other words, $\rho_2(x_2) = \rho_2(e_3) + \rho_3
(x_2)$. On the other hand, if $x_2 < e_3$, then $\#[e_2, e_3[\ = \#[e_2, x_2[\ +\
\#[x_2, e_3[$, in other words, $\rho_2(e_3) = \rho_2(x_2) - \rho_3(x_2)$. Therefore,
\begin{equation}
\label{eq:r3r2}
x_2\in G_2 \quad\Longrightarrow\quad \rho_3(x_2) = \rho_2(x_2) - \rho_2(e_3).
\end{equation}

The next two propositions prove the correctness of our Gregorian calendar
algorithms.  \cref{prp:RataDieGregorian} considers the calculation of the rata
die function for any chosen epoch and \cref{prp:RataDieInverseGregorian}
considers its inverse.

\begin{proposition}
\label{prp:RataDieGregorian}
Let $\rho_2$, $\rho_3$ be rata die functions on $G_2$ with epochs $e_2 = (z_2,
3, 1)$ and $e_3$, respectively, where $z_2\in\Z$ is a multiple of $400$. Given
$(y_2, m_2, d_2)\in G_2$, set:
\begin{equation*}
\begin{aligned}
y_1 &= y_2 - z_2, \\
m_1 &= m_2, \\
d_1 &= d_2, \\
\end{aligned}
\quad\vrule\quad
\begin{aligned}
y_0  &= y_1 - {\bf 1}_{\{m_1\le 2\}}, \\
m_0 &= m_1 + 12\cdot{\bf 1}_{\{m_1\le 2\}}, \\
d_0 &= d_1 - 1,
\end{aligned}
\quad\vrule\quad
\begin{aligned}
q_1 &= y_0/100, \\
y_c &= 1461\cdot y_0/4 - q_1 + q_1/4, \\
m_c &= (979\cdot m_0 - 2919)/2^5, \\
d_c &= d_0.
\end{aligned}
\end{equation*}
Then, $\rho_3(y_2, m_2, d_2) = y_c + m_c + d_c - \rho_2(e_3)$.
\end{proposition}

\begin{proof}
The first column and \cref{eq:r2r1} give $(y_1, m_1,
d_1)\in G_1$ and $\rho_2(y_2, m_2, d_2) = \rho_1(y_1, m_1, d_1)$. The second column
gives $(y_0, m_0, d_0) = P_1(y_1, m_1, d_1)$ and \cref{eq:r1r0} yields $(y_0,
m_0, d_0) \in G_0$ and $\rho_1(y_1, m_1, d_1) = \rho_0(y_0, m_0, d_0)$. Hence,
$\rho_2(y_2, m_2, d_2) = \rho_0(y_0, m_0, d_0)$, and from \cref{cor:RataDie} we
obtain $\rho_2(y_2, m_2, d_2) = y_c + m_c + d_c$. The result follows from
\cref{eq:r3r2}.
\end{proof}

\begin{proposition}
\label{prp:RataDieInverseGregorian}
Let $\rho_2$, $\rho_3$ be rata die functions on $G_2$ with epochs $e_2 = (z_2,
3, 1)$ and $e_3\ge e_2$, respectively, where $z_2\in\Z$ is a multiple of $400$.
Given $r\in\Z$ with $r\ge -\rho_2(e_3)$, set $r_0 = r + \rho_2(e_3)$ and:
\begin{equation*}
\begin{aligned}
n_1 &= 4\cdot r_0 + 3, \\
\\
q_1 &= n_1/146097, \\
r_1 &= n_1\%146097/4, \\
\end{aligned}
\quad\vrule\quad
\begin{aligned}
n_2 &= 4\cdot r_1 + 3, \\
u_2 &= 2\s939\s745\cdot n_2, \\
q_2 &= u_2/2^{32}, \\
r_2 &= u_2\%2^{32}/2\s939\s745/4, \\
\end{aligned}
\quad\vrule\quad
\begin{aligned}
n_3 &= 2141\cdot r_2 + 197913, \\
\\
q_3 &= n_3/2^{16}, \\
r_3 &= n_3\%2^{16}/2141. \\
\end{aligned}
\end{equation*}
and:
\begin{equation}
\label{eq:Bottom}
\begin{aligned}
y_0 &= 100\cdot q_1 + q_2, \\
m_0 &= q_3, \\
d_0 &= r_3, \\
\end{aligned}
\quad\vrule\quad
\begin{aligned}
y_1 &= y_0 + {\bf 1}_{\{r_2\ge 306\}}, \\
m_1 &= m_0 - 12\cdot{\bf 1}_{\{r_2\ge 306\}}, \\
d_1 &= d_0 + 1. \\
\end{aligned}
\quad\vrule\quad
\begin{aligned}
y_2 &= y_1 + z_2, \\
m_2 &= m_1, \\
d_2 &= d_1.
\end{aligned}
\end{equation}
Then, $\rho_3(y_2, m_2, d_2) = r$.
\end{proposition}

\begin{proof}
Since $r\ge -\rho_2(e_3)$, we have $r_0 = r + \rho_2(e_3)\ge 0$ and thus
$x_0\in G_0$ exists such that $\rho(x_0) = r_0$.

\Cref{prp:Century} provides that $q_1$ is the century of $x_0$ and that $r_1$ is its day
of the century. \cref{cor:Year} then yields that $q_2$ is the year of the century of
$x_0$ and $r_2$ is its day of the year, while \cref{cor:Month} states that $q_3$ is the
month of $x_0$ and $r_3$ is its day.

Hence, the first column of \cref{eq:Bottom} reads $x_0 = (y_0, m_0, d_0)$. From
\cref{rmk:JanOrFeb} we can substitute ${\bf 1}_{\{r_2\ge 306\}}$ with ${\bf 1}_
{\{m_0\ge 13\}}$ and then, by \cref{eq:P1Inverse} the second column of
\cref{eq:Bottom}
becomes $(y_1, m_1, d_1) = P_1^{-1}(x_0)$ and \cref{eq:r0r1} gives $(y_1, m_1,
d_1)\in G_1$ and $\rho_1(y_1, m_1, d_1) = \rho_0(x_0)$. Finally, \cref{eq:r1r2}
and the third column of \cref{eq:Bottom} yield $(y_2, m_2, d_2)\in G_2$ and $\rho_2(y_2, m_2,
d_2) = \rho_1(y_1, m_1, d_1) = \rho_0(x_0) = r_0 = r + \rho_2(e_3)$.

We have shown that $r = \rho_2(y_2, m_2, d_2) - \rho_2(e_3)$. Therefore,
\cref{eq:r3r2} gives $\rho_3(y_2, m_2, d_2) = r$.
\end{proof}

\section{Performance analysis}
\label{sec:Performance}

We benchmarked our algorithms from \cref{prp:RataDieGregorian,%
prp:RataDieInverseGregorian} against counterparts in five of the most widely used C,
C++, C\# and Java libraries, as listed below:

\begin{tabular}{p{.1\textwidth}p{.84\textwidth}}
{\bf glibc} & The GNU C Library \cite{GlibcMktime, GlibcOfftime}. (The Linux
Kernel contains a similar implementation \cite{LinuxTimeconv}.)
\\
{\bf Boost} & The Boost C++ libraries \cite{BoostGregorianCalendar}.
\\
{\bf libc++} & LLVM's implementation of the C++ Standard Library
\cite{LLVMChrono}.
\\
{\bf .NET} & Microsoft .NET framework \cite{DotNetDateTime}.
\\
{\bf OpenJDK} & Oracle's open source implementation of the Java Platform SE
\cite{OpenJDKLocalDate}. (Android uses the same code \cite{AndroidLocalDate}.)
\end{tabular}

We used source files as publicly available on 2 May 2020. Non-C++
implementations have been ported to this language and have all been slightly
modified to achieve consistent (a) function signatures; (b) storage types (for
years, months, days and rata dies); and (c) epoch (Unix epoch, {\em i.e.}, 1
January 1970). Some originals deal with date and time but our variants work on
dates only. (Given the uniform durations of days, hours, minutes and seconds,
it would be trivial to incorporate the time component to any dates-only algorithm.
Moreover, \Cref{ex:FastTime} suggests improvements that are not used by major
implementations.)

We did not include Microsoft's C++ Standard Library because in May 2020 it did not yet
implement these functionalities. For the same reason, libstdc++, the
GNU implementation of the C++ Standard Library, is also absent. Furthermore, a
forthcoming release of libstdc++ is expected to implement our algorithms.

We also considered our own implementations of algorithms described in the academic
literature, namely, Baum \cite{Baum1998}, Fliegel and Flandern
\cite{FliegelFlandern1968}, Hatcher \cite{Hatcher1984, Hatcher1985,
Richards1998} and Reingold and Dershovitz \cite{ReingoldDershowitz2018}.

Timings were obtained with the help of the Google Benchmark library
\cite{GoogleBenchmark}, to which we delegated the task of producing
statistically relevant results. The code was run on an Intel i7-10510U CPU at
$4.9$ GHz. Source code, available at \cite{NeriSchneider2020_code}, was compiled
by GCC 10.2.0 at optimisation level \texttt{-O3}. Results vary across platforms
but maintain qualitative consistency between them.

The table in \cref{fig:RataDie} shows the time taken by each algorithm to
evaluate $\rho_3$ at $16\s384$ pseudo-random dates, uniformly distributed in
$[(1570, 1, 1), (2370, 1, 1)[$ ({\em i.e.,} Unix epoch $\pm 400$ years). They
encompass the time spent scanning the array of dates (also shown). Subtracting
the scanning time from that of each algorithm gives a fairer account of the time spent
by the algorithm itself. The chart plots these adjusted timings relative to
ours (\cref{prp:RataDieGregorian}).

\begin{figure}[ht]
\small
\includegraphics[width=.678\textwidth]{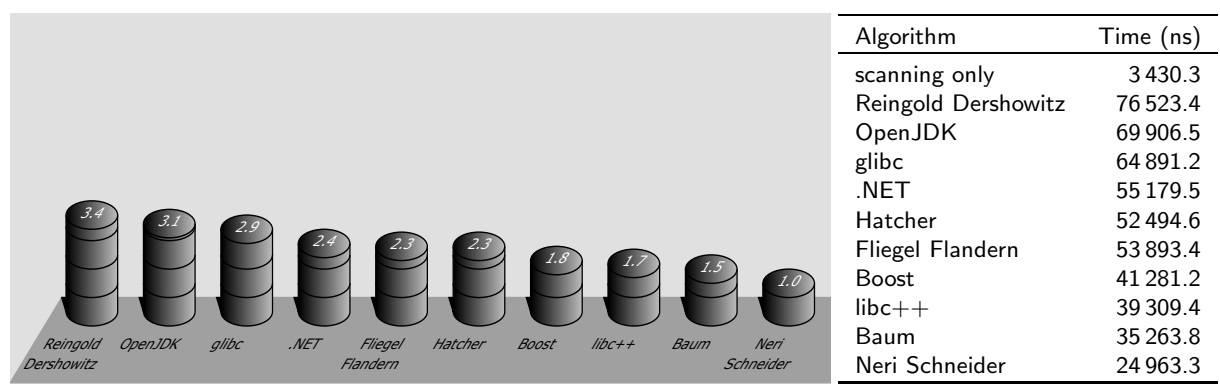}
\begin{tabular}[b]{lr}
\hline\noalign{\smallskip}
Algorithm           & Time (ns) \\
\hline\noalign{\smallskip}
scanning only       &  3\s430.3 \\
Reingold Dershowitz & 76\s523.4 \\
OpenJDK             & 69\s906.5 \\
glibc               & 64\s891.2 \\
.NET                & 55\s179.5 \\
Hatcher             & 52\s494.6 \\
Fliegel Flandern    & 53\s893.4 \\
Boost               & 41\s281.2 \\
libc++              & 39\s309.4 \\
Baum                & 35\s263.8 \\
Neri Schneider      & 24\s963.3 \\
\hline
\end{tabular}%
\caption{Relative and absolute timings of rata die evaluations.}
\label{fig:RataDie}
\end{figure}

Similarly, the table in \cref{fig:RataDieInverse} shows the time taken by each
algorithm to evaluate $\rho_3^{-1}$ at $16384$ pseudo-random integer numbers
uniformly distributed in $[-146097, 146097[$ (again, Unix epoch $\pm 400$ years
as per \cref{rmk:LeapCycles}). The chart displays adjusted times relative to
ours (\cref{prp:RataDieInverseGregorian}).

\begin{figure}[ht]
\small
\includegraphics[width=.678\textwidth]{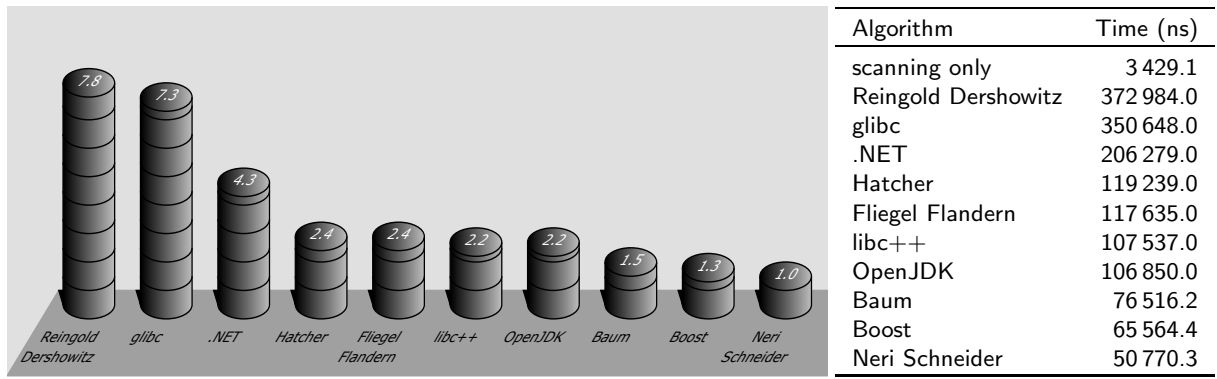}
\begin{tabular}[b]{lr}
\hline\noalign{\smallskip}
Algorithm           & Time (ns) \\
\hline\noalign{\smallskip}
scanning only       &   3\s429.1 \\
Reingold Dershowitz & 372\s984.0 \\
glibc               & 350\s648.0 \\
.NET                & 206\s279.0 \\
Hatcher             & 119\s239.0 \\
Fliegel Flandern    & 117\s635.0 \\
libc++              & 107\s537.0 \\
OpenJDK             & 106\s850.0 \\
Baum                &  76\s516.2 \\
Boost               &  65\s564.4 \\
Neri Schneider      &  50\s770.3 \\
\hline
\end{tabular}%
\caption{Relative and absolute timings of rata die inverse evaluations.}%
\label{fig:RataDieInverse}
\end{figure}

\bibliographystyle{plain}
\bibliography{../bibtex/articles,../bibtex/books,../bibtex/websites}

\section*{Disclaimer}

{\scriptsize
Opinions and estimates constitute our judgement as of the date of this Material,
are for informational purposes only and are subject to change without notice.
This Material is not the product of J.P. Morgan's Research Department and
therefore, has not been prepared in accordance with legal requirements to
promote the independence of research, including but not limited to, the
prohibition on the dealing ahead of the dissemination of investment research.
This Material is not intended as research, a recommendation, advice, offer or
solicitation for the purchase or sale of any financial product or service, or to
be used in any way for evaluating the merits of participating in any
transaction. It is not a research report and is not intended as such. Past
performance is not indicative of future results. Please consult your own
advisors regarding legal, tax, accounting or any other aspects including
suitability implications for your particular circumstances. J.P. Morgan
disclaims any responsibility or liability whatsoever for the quality, accuracy
or completeness of the information herein, and for any reliance on, or use of
this material in any way. Important disclosures at: www.jpmorgan.com/disclosures
}

\end{document}